\documentclass[11pt]{llncs}

\usepackage[margin=1in]{geometry}
\usepackage{url}
\usepackage{color}
\usepackage{times}
\usepackage{latexsym}
\usepackage{graphicx}
\usepackage{amsmath}
\usepackage{epsfig}
\usepackage{amssymb}
\usepackage{xspace}
\usepackage[ruled,noend]{algorithm2e}
\usepackage{subfigure}
\usepackage{verbatim}
\usepackage{ulem} 

\usepackage{booktabs}
\usepackage{multirow}
\usepackage{multicol}
\usepackage{colortbl}
\usepackage{rotating}

\newcommand{\myboldmath}{}%\boldmath doesn't work with latex
\newcommand{\defn}[1]           {{\textit{\textbf{\myboldmath #1}}}}

\newcommand{\parhead}[1]{\textbf{#1.}\xspace}

\newtheorem{observation}{Observation}

\newcommand{\VMA}{{$(\cdot,\cdot)$-VMA}\xspace}

\renewcommand{\P}{{\pi}\xspace}
\renewcommand{\l}{{\ell}\xspace}
\newcommand{\alg}{{ALG}\xspace}
\newcommand{\opt}{{OPT}\xspace}

\newcommand{\CMVMA}{{$(C,m)$-VMA}\xspace}
\newcommand{\MVMA}{{$(\cdot,m)$-VMA}\xspace}
\newcommand{\CVMA}{{$(C,\cdot)$-VMA}\xspace}
\renewcommand{\L}{{\ell(D)}\xspace}
\newcommand{\OPTBP}{\ensuremath{\overline{m}}\xspace}

\newcommand{\remove}[1]{}

\newcommand{\lin}[1]{\textcolor{black}{#1}}

\newcommand{\jordi}[1]{\textcolor{black}{#1}}

\begin{document}

% Page heads
\markboth{J. Arjona Aroca, A. Fern\'andez Anta, M. A. Mosteiro, C. Thraves and L. Wang}{Power-efficient Assignment of Virtual Machines to Physical Machines}

% Title portion
\title{Power-efficient Assignment of Virtual Machines to Physical Machines
\thanks{This work has been supported in part by the 
Comunidad de Madrid grant S2009TIC-1692, 
the MINECO grant TEC2011-29688-C02-01, 
the National Natural Science Foundation of China grant 61020106002, 
the MICINN grant Juan de la Cierva,
the National Science Foundation (CCF-0937829, CCF-1114930), 
and Kean University UFRI grant.}
}

\author{Jordi Arjona Aroca\inst{1,2} \and Antonio Fern\'andez Anta\inst{1}\and Miguel A. Mosteiro\inst{3,4}\and\\ Christopher Thraves\inst{4} \and Lin Wang\inst{5,6}}

\institute{Institute IMDEA Networks, Madrid, Spain. \email{\{jorge.arjona,antonio.fernandez\}@imdea.org}
\and
Universidad Carlos III de Madrid, Madrid, Spain.
\and
Department of Computer Science, Kean University, Union, NJ, USA. \email{mmosteir@kean.edu}
\and
CNRS-LAAS and Univ. Toulouse - LAAS, Tolouse, France. \email{cbthraves@gsyc.es}
\and
Institute of Computing Technology, Chinese Academy of Sciences, Beijing, China. \email{wanglin@ict.ac.cn}
\and
University of Chinese Academy of Sciences, Beijing, China.
%\and
%MICINN Juan de la Cierva Researcher at GSyC, Universidad Rey Juan Carlos,  Spain
}

%%%%%%% reduced author list to save space if necessary
%\author{Jordi Arjona Aroca~\thanks{Institute IMDEA Networks \& Univ. Carlos III de Madrid, Madrid, Spain.\texttt{jorge.arjona@imdea.org}}~
%Antonio Fern\'andez Anta~\thanks{Institute IMDEA Networks, Madrid, Spain.
%\texttt{antonio.fernandez@imdea.org}}~
%Miguel A. Mosteiro~\thanks{Computer Science Dept., Kean Univ., Union, NJ, USA \& GSyC, Univ. Rey Juan Carlos, Spain.
%\texttt{mmosteir@kean.edu}}~
%and~Christopher Thraves~\thanks{GSyC, Univ. Rey Juan Carlos,  Spain.
%\texttt{cbthraves@gsyc.es}}
%}

\maketitle

\begin{abstract}
Motivated by current trends in cloud computing, we study a version of the generalized assignment problem where a set of virtual processors has to be implemented by a set of \emph{identical} processors. 
For literature consistency, we say that a set of virtual machines (VMs) is assigned to a set of physical machines (PMs). The optimization criteria is to minimize the power consumed by all the PMs. We term the problem Virtual Machine Assignment (VMA). Crucial differences with previous work include a variable number of PMs, that each VM must be assigned to exactly one PM (i.e., VMs cannot be implemented fractionally), and a minimum power consumption for each active PM. 
Such infrastructure may be strictly constrained in the number of PMs or in the PMs' capacity, depending on how costly (in terms of power consumption) it is to add a new PM to the system or to heavily load some of the existing PMs.
Low usage or ample budget yields models where PM capacity and/or the number of PMs may be assumed unbounded for all practical purposes.
We study four VMA problems depending on whether the capacity or the number of PMs is bounded or not.
Specifically, we study %NP-completeness, offline approximability, 
hardness
and online competitiveness for a variety of cases.
To the best of our knowledge, this is the first comprehensive study of the VMA problem for this cost function.
\end{abstract}

\keywords{Cloud computing, generalized assignment, scheduling, load balancing.}

\newpage

\section{Introduction}
%!TEX root = ./main.tex

The current pace of technology developments,
and the continuous change in business requirements,
may rapidly yield a given
proprietary computational platform obsolete, oversized, or insufficient.
Thus, outsourcing has recently become a popular approach to obtain computational services without incurring in amortization costs.
Furthermore, in order to attain flexibility, such service is usually virtualized, so that the user may tune the computational platform to its particular needs.
Users of such service need not to be aware of the particular implementation, 
%but only of the specification of the virtual machine they use.
they only need to specify the virtual machine they want to use.
This conceptual approach to outsourced computing has been termed \emph{cloud computing}, in reference to the cloud symbol used as an abstraction of a complex infrastructure in system diagrams.
%But also to convey the irrelevance (from the perspective of the user) of the specifics of what is hidden inside ``the cloud''.
Current examples of cloud computing providers include Amazon Web Services~\cite{ec2}, Rackspace~\cite{rackspace}, and Citrix~\cite{citrix}.

Depending on what the specific service provided is, the cloud computing model comes in different flavors, such as \emph{infrastructure as a service}, \emph{platform as a service}, \emph{storage as a service}, etc. In each of these models, the user may choose specific parameters of the computational resources provided. For instance, processing power, memory size, communication bandwidth, etc. Thus, in a cloud-computing service platform, various \defn{virtual machines (VM)} with user-defined specifications must be implemented by, or \defn{assigned to}\footnote{The cloud-computing literature use instead the term \emph{placement}. We choose here the term assignment for consistency with the literature on general assignment problems.},
 various \defn{physical machines (PM)}\footnote{We choose the notation VM and PM for simplicity and consistency, but notice that our study applies to any computational resource assignment problem, as long as the minimization function is the one modeled here.}.
Furthermore, such a platform must be scalable, allowing to add more PMs, should the business growth require such expansion. %The number of PMs available may be fixed, but a setting where PMs may be added to the system if convenient is also of interest.
In this work, we call this problem the \defn{Virtual Machine Assignment (VMA)} problem. 
%(A precise definition is given in Section~\ref{sec:problem}.)

The optimization criteria for VMA depends on what the particular objective function sought is.
From the previous discussion, it can be seen that, underlying VMA, there is some form of bin-packing problem. However, in VMA the number of PMs  (i.e., bins for bin packing) may be increased if needed. \lin{Since CPU is generally the dominant power consumer in a server \cite{eenergy}, VMA is usually carried out according to CPU workloads. With only the static power consumption of servers considered, previous work related to VMA has focused on minimizing the number of active PMs (cf.~\cite{BRKplacement} and the references therein) in order to minimize the total static energy consumption. This is commonly known as VM consolidation \cite{Nathuji_Schwan-2007,Kusic_Kephart-2009}. However, despite the static power, the dynamic power consumption of a server, which has been shown to be superlinear on the \defn{load} of a given computational resource \cite{PruhsSpeedScaling,PruhsPrimalDual}, is also significant and cannot be ignored. Since the definition of load is not precise, we borrow the definition in \cite{eenergy} and define the load of a server as the amount of active cycles per second a task requires, an absolute metric independent of the operating frequency or the number of cores of a PM. The superlinearity property of the dynamic power consumption is also confirmed by the results in \cite{eenergy}. As a result, when taking into account both parts of power consumption, the use of extra PMs may be more efficient energy-wise than a minimum number of heavily-loaded PMs. This inconsistency with the literature in VM consolidation has been supported by the results presented in \cite{eenergy} and, hence, we claim that the way consolidation has been traditionally performed has to be reconsidered. In this work, we combine both power-consumption factors and explore the most energy-efficient way for VMA. That is, for some parameters $\alpha>1$ and $b>0$, we seek to minimize the sum of the $\alpha$ powers of the PMs loads \emph{plus} the fixed cost $b$ of using each PM.}

Physical resources are physically constrained.
A PMs infrastructure may be strictly constrained in the number of PMs or in the PMs \jordi{CPU capacity}.
However, if usage patterns indicate that the PMs will always be loaded well below their capacity, it may be assumed that the capacity is unlimited. Likewise, if the power budget is very big, the number of PMs may be assumed unconstrained for all practical purposes. 
These cases yield 4 VMA subproblems, depending on whether the capacity and the number of PMs is limited or not.
We introduce these parameters denoting the problem as {\bf({\em C,m})-VMA}, where $C$ is the PM \jordi{CPU capacity}, $m$ is the maximum number of PMs, and each of these parameters is replaced by a dot if unbounded. 

In this work, we study the hardness and online competitiveness of the VMA problem.
Specifically, we show that VMA is NP-hard \emph{in the strong sense} (in particular, we observe that $(C, m)$-VMA is strongly NP-complete). Thus, VMA problems do not have a fully polynomial time approximation scheme (FPTAS). Nevertheless, using previous results derived for more general objective functions, we notice that $(\cdot,m)$- and $(\cdot,\cdot)$-VMA have a polynomial time approximation scheme (PTAS). We also show various lower and upper bounds on the offline approximation and the online competitiveness of VMA. Rather than attempting to obtain tight bounds for particular instances of the parameters of the problem ($C,m,\alpha,b$) we focus on obtaining \emph{general bounds}, whose parameters can be instantiated for the specific application. The bounds obtained show interesting trade-offs between the PM capacity and the fixed cost of adding a new PM to the system.
To the best of our knowledge, this is the first VMA study that is focused on power consumption.

\parhead{Roadmap} The paper is organized as follows. 
In what remains of this section, we define formally the \VMA problem, we overview the related work, and we describe our results in detail. Section~\ref{sec:prelim} includes some preliminary results that will be used throughout the paper. The offline and online analyses are included in Section~\ref{sec:alloffline} and~\ref{sec:comp} respectively. \lin{Section~\ref{sec:discussion} discusses some practical issues and provides some useful insights regarding real implementation.} For succinctness, many of the
proofs are left to the full version of this paper in \cite{fullThisPaper}.

    \subsection{Problem Definition}
    \label{sec:problem}
    %We define the \defn{Virtual Machine Assignment problem,} with Bounded Capacity and Bounded number of Physical Machines VMA.
We describe the \VMA problem now. Given a set $S=\{s_1, \ldots , s_m\}$ of $m>1$ identical physical machines (PMs) of capacity $C$; rational numbers $\mu$, $\alpha$ and $b$, where $\mu>0$, $\alpha >1$ and $b>0$; a set $D=\{d_1, \ldots, d_n\}$ of $n$ virtual machines and a function $\ell: D \rightarrow \mathbb{R}$ that gives the CPU load each virtual machine incurs\footnote{For convenience, 
 we overload the function $\l(\cdot)$ to be applied over sets of virtual machines, so that for any set $A \subseteq D, \l(A)=\sum_{d_j \in A} \l(d_j)$.},
we aim to obtain a partition $\P=\{A_1, \ldots, A_m\}$ of $D$, such that $\l(A_i) \leq C$, for all $i$.
Our objective will be then minimizing the power consumption given by the function
\begin{equation}
\label{eq:model}
P(\P)=\sum_{i \in [1,m]:A_i \neq \emptyset} \Bigg( \mu \Big(\sum_{d_j \in A_i} \l(d_j) \Big)^\alpha + b \Bigg).
\end{equation}
Let us define the function $f(\cdot)$, such that $f(x)=0$ if $x=0$ and $f(x)=\mu x^\alpha + b$ otherwise. Then, the objective function is to minimize
$
P(\P)=\sum_{i=1}^m f(\l(A_i)).
$
The parameter $\mu$ is used for consistency with the literature. For clarity we will consider $\mu=1$ in the rest of the paper.
All the results presented apply for other values of $\mu$.

We also study several special cases of the VMA problem, namely \CMVMA, \CVMA, \MVMA and \VMA. \CMVMA refers to the 
case where both the number of available PMs and its capacity are fixed.
\VMA, where $(\cdot)$ denotes unboundedness, refers to the case where both the number of available PMs and its capacity are unbounded (i.e., $C$ is larger than the total load of the VMs that can ever be in the system at any time, or $m$ is larger than the number of VMs that can ever be in the system at any time). \CVMA and \MVMA are the cases where the number of available PMs and their capacity is unbounded, respectively.

\begin{comment}
We define the \defn{Virtual Machine Assignment (VMA)} problem as follows:
\begin{description}
\item[Input:]
%\begin{itemize}
%\item
A set $S=\{s_1, \ldots , s_m\}$ of $m$ identical physical machines (PMs)\footnote{We will consider the PMs with bounded or unbounded capacity in the following analysis}.
%\item
Rational numbers $\mu$, $\alpha$ and $b$, where $\mu>0$, $\alpha >1$ and $b>0$.
%\item
A set $D=\{d_1, \ldots, d_n\}$ of $n$ virtual machines.
%\item
A function
%$\l(\cdot)$
$\ell: D \rightarrow \mathbb{R}$
 that gives the CPU load each virtual machine incurs.
%\end{itemize}

\item[Output:]
A partition $\P=\{A_1, \ldots, A_m\}$ of $D$.
% that minimizes the power consumption given by the function
%$$
%P(\P)=\sum_{i \in [1,m]:A_i \neq \emptyset} ((\sum_{d_j \in A_i} \l(d_j))^\alpha + b)
%$$

\item[Objective function:]
Minimize the power consumption given by the function
\begin{equation}
\label{eq:model}
P(\P)=\sum_{i \in [1,m]:A_i \neq \emptyset} \Bigg( \Big(\mu\sum_{d_j \in A_i} \l(d_j) \Big)^\alpha + b \Bigg).
\end{equation}
\end{description}

For convenience, we define the following notation.
We overload the function $\l(\cdot)$ to be applied over sets of virtual machines, so that $\l(A_i)=\sum_{d_j \in A_i} \l(d_j)$.
Also, let us define the function $f(\cdot)$, such that $f(x)=0$ if $x=0$ and $f(x)=x^\alpha + b$ otherwise. Then, the objective function is to minimize
$$
P(\P)=\sum_{i=1}^m f(\l(A_i)).
$$
\end{comment}

    \subsection{Related Work}
    \label{sec:relwork}
    %!TEX root = ./main.tex

To the best of our knowledge, previous work on VMA has been only experimental~\cite{cloudconsolidation,KDRassign,MFDcloud,cloudheuristics} or has focused on different cost functions~\cite{DBLP:journals/jacm/CodyC76,DBLP:conf/soda/AlonAWY97,BRKplacement,DBLP:journals/siamcomp/ChandraW75}.
First, we provide an overview of previous theoretical work for related assignment problems (storage allocation, scheduling, network design, etc.). The cost functions considered in that work resemble or generalize the power cost function under consideration here. Secondly, we overview related experimental work.

%%%%%theoretical

Chandra and Wong \cite{DBLP:journals/siamcomp/ChandraW75}, and Cody and Coffman \cite{DBLP:journals/jacm/CodyC76} study a problem for storage allocation that is a variant of \MVMA with $b=0$ and $\alpha=2$. Hence, this problem tries to minimize the sum of the squares of the machine-load vector for a fixed number of machines. They study the offline version of the problem and provide algorithms with constant approximation ratio. A significant leap was taken by Alon et al.~\cite{DBLP:conf/soda/AlonAWY97}, since they present a PTAS for the problem of minimizing the $L_p$ norm of the
load vector, for any  $p\geq 1$. This problem has the previous one as special case, and is also a variant of the \MVMA problem when $p=\alpha$ and $b=0$. Similarly, Alon et al~\cite{alon1998approximation} extended this work for a more general set of functions, that include $f(\cdot)$ as defined above.
Hence, their results can be directly applied in the \MVMA problem. Later, Epstein et al~\cite{epstein2004} extended~\cite{alon1998approximation} further for the uniformly related machines case. We will use these results in Section \ref{sec:alloffline} in the analysis of the offline case of \MVMA and \VMA.

Bansal, Chan, and Pruhs minimize arbitrary power functions for speed scaling in job scheduling~\cite{PruhsSpeedScaling}.
The problem is to schedule the execution of $n$ computational jobs on a \emph{single} processor, whose speed may vary within a countable collection of intervals. Each job has a release time, a processing work to be done, a weight characterizing its importance, and its execution can be suspended and restarted later without penalty. A scheduler algorithm must specify, for each time, a job to execute and a speed for the processor. The goal is to minimize the weighted sum of the flow times over all jobs plus the energy consumption, where the flow time of a job is the time elapsed from release to completion and the energy consumption is given by $s^\alpha$ where $s$ is the processor speed and $\alpha>1$ is some constant.
For the online algorithm \emph{shortest remaining processing time first}, the authors prove a $(3+\epsilon)$ competitive ratio for the objective of total weighted flow plus energy. Whereas for the online algorithm \emph{highest density first (HDF)}, where the density of a job is its weight-to-work ratio, they prove a $(2+\epsilon)$ competitive ratio for the objective of fractional weighted flow plus energy.

Recently, Im, Moseley, and Pruhs studied online scheduling for general cost functions of the flow time, with the only restriction that such function is non-decreasing~\cite{PruhsGralCostFunc}. In their model, a collection of jobs, each characterized by a release time, a processing work, and a weight, must be processed by a \emph{single} server whose speed is variable. A job can be suspended and restarted later without penalty.
The authors show that HDF is $(2+\epsilon)$-speed $O(1)$-competitive against the optimal algorithm on a unit speed-processor, for all non-decreasing cost functions of the flow time.
Furthermore, they also show that this ratio cannot be improved significantly proving impossibility results if the cost function is not uniform among jobs or the speed cannot be significantly increased.
%\textcolor{red}{More references here to previous work with specific cost functions, such as polynomial. Also an offline algorithm with a $O(\log \log n P)$ approximation, where $P$ is the maximum job size.}

A generalization of the above problem is studied by Gupta, Krishnaswamy, and Pruhs in~\cite{PruhsPrimalDual}. The question addressed is how to assign jobs, \emph{possibly fractionally}, to unrelated parallel machines in an online fashion in order to minimize the sum of the $\alpha$-powers of the machine loads plus the assignment costs. Upon arrival of a job, the algorithm learns the increase on the load and the cost of assigning a unit of such job to a machine. Jobs cannot be suspended and/or reassigned. The authors model a greedy algorithm that assigns a job so that the cost is minimized as solving a mathematical program with constraints arriving online. They show a competitive ratio of $\alpha^\alpha$ with respect to the solution of the dual program which is a lower bound for the optimal.
They also show how to adapt the algorithm to integral assignments with a $O(\alpha)^\alpha$ competitive ratio, which applies directly to our \MVMA problem.
References to previous work on the particular case of minimizing energy with deadlines can be found in this paper.
%\textcolor{red}{Lots of references here to minimizing energy with deadlines, and why it is a particular case of this problem.}

%Similar cost functions have been considered for the minimum cost network-design problem. In this problem packets have to be routed through a (possibly multihop) network of speed scalable routers. There is a cost associated to assigning a packet to a link and to the speed or load of the router. The goal is to route all packets minimizing the aggregated cost.
%In~\cite{DBLP:journals/ton/AndrewsAZZ12} and
%\cite{DBLP:conf/focs/AndrewsAZ10} the authors show offline algorithms for this problem that achieve polynomial and poly-logarithmic approximation, respectively, where the cost function is the $\alpha$-th power of the link load plus a link assignment cost, for any \emph{constant} $\alpha>1$. The same problem and cost function is studied in~\cite{PruhsPrimalDual} (the assignment cost is omitted for clarity). As for the scheduling problem, the authors model a greedy algorithm as solving a mathematical program with constraints arriving online. They show a competitive ratio of $\alpha^\alpha$ with respect to the solution of the dual program which is a lower bound for the optimal.

Similar cost functions have been considered for the minimum cost network-design problem. In this problem, packets have to be routed through a (possibly multihop) network of speed scalable routers. There is a cost associated to assigning a packet to a link and to the speed or load of the router. 
The goal is to route all packets minimizing the aggregated cost. 
In~\cite{DBLP:journals/ton/AndrewsAZZ12} and \cite{DBLP:conf/focs/AndrewsAZ10} the authors show offline algorithms for this problem with undirected graph and homogeneous link cost functions that achieve polynomial and poly-logarithmic approximation, respectively. The cost function is the $\alpha$-th power of the link load plus a link assignment cost, for any constant $\alpha > 1$. 
The same problem and cost function is studied in \cite{PruhsPrimalDual}.
%
%%%%%%%%%%%%CROSSED-OUT TEXT
%\sout{(the assignment cost is omitted for clarity). As for the scheduling problem, the authors model a greedy algorithm as solving a mathematical program with constraints arriving online. They show a competitive ratio of $\alpha^\alpha$ with respect to the solution of the dual program which is a lower bound for the optimal. }
%%%%%%%%%%%%CROSSED-OUT TEXT
%
Bansal \emph{et al.} \cite{DBLP:conf/medalg/BansalGKNPS12} study a minimum-cost virtual circuit multicast routing problem with speed scalable links. They give a polynomial-time $O(\alpha)$-approximation offline algorithm and a 
polylog-competitive online algorithm, both for the case with homogeneous power functions. They also show that the problem is APX-hard in the case with heterogeneous power functions and there is no polylog-approximation when the graph is directed. \lin{Recently, Antoniadis \emph{et al.} \cite{AntoniadisVCR} improved the results by providing a simple combinatorial algorithm that is $O(\log^{\alpha}n)$-approximate, from which we can construct an $\widetilde{O}(\log^{3\alpha + 1}n)$-competitive online algorithm.}
The \MVMA problem can be seen as a especial case of the problem considered in these papers in which there are only two nodes, source and destination, and $m$ parallel links connecting them.

To the best of our knowledge, the problem of minimizing the power consumption (given in~Eq.\ref{eq:model}) with capacity constraints (i.e., the \CMVMA and \CVMA problems) has received very limited attention, in the realm of both VMA and network design, although the approaches in \cite{DBLP:conf/focs/AndrewsAZ10} and \cite{DBLP:conf/medalg/BansalGKNPS12} are related to or based on the solutions for the capacitated network-design problem \cite{DBLP:conf/ipco/ChakrabartyCKK11}.

%%%experimental

%\color{blue}

The experimental work related to VMA is vast and its detailed overview is out of the scope of this paper.
Some of this work does not minimize energy~\cite{CSPChMapReduceCloud,MKMfaultTolerant,MSvector} or it applies to a model different than ours (VM migration~\cite{NDMautonomic,SKZconsolidation}, knowledge of future load~\cite{MNTevolutionary,SKZconsolidation}, feasibility of allocation~\cite{BRKplacement}, multilevel architecture~\cite{MFDcloud,NDMautonomic,JBcloud}, interconnected VMs~\cite{interVMs}, etc.). On the other hand, some of the experimental work where minimization of energy is evaluated focus on a more restrictive cost function~\cite{VLRPPGappCentric,JBcloud,XFmultiobjective}.

%\jordi{In~\cite{JBcloud},
%for an energy cost model that is linear,
%the authors evaluate experimentally the allocation of VMs to clusters following 7 placement policies,
%some of them included in popular cloud platforms~\cite{OpenNebula,Eucalyptus}.
%Namely, Round Robin, Striping, Packing, Load Balancing (free CPU count), Load Balancing (free CPU ratio), Watts per Core, Cost per Core. We adapt 5 of these policies (defined later) to our model and cost function for the purpose of simulations.}

In~\cite{SKZconsolidation}, the authors focus on an energy-efficient VM placement problem with two requirements: CPU and disk. These requirements are assumed to change dynamically and the goal is to consolidate loads among servers, possibly using migration at no cost.
In our model VMs assignment is based on a CPU requirement that does not change and migration is not allowed. Should any other resource be the dominating energy cost, the same results apply for that requirement. Also, if loads change and migration is free, an offline algorithm can be used each time that a load changes or a new VM arrives.
In~\cite{SKZconsolidation} it is shown experimentally that energy-efficient VMA does not merely reduce to a packing problem. That is, to minimize the number of PMs used even if their load is close to their maximum capacity. For our model, we show here that the optimal load of a given server is a function only of the fixed cost of being active ($b$) and the exponential rate of power increase on the load ($\alpha$). That is, the optimal load is not related to the maximum capacity of a PM.

\color{black}

    \subsection{Our Results}
    \label{sec:results}
    %\input{np}
    %!TEX root = ./main.tex

%!TEX root = ./main.tex
\renewcommand{\tabcolsep}{1pt}
\begin{table}[t]
\centering
\small

\begin{tabular}{|l|cc|cc|}
\hline %%%%%%%%%%%%%%%%%%%%%%%%%%%%%%%%%%%%%%%%%%%%%%%%%%%
\rule{0pt}{3ex}
\begin{tabular}{l}VMA\\subprob.\end{tabular} & \multicolumn{2}{c|}{$x^*<C$} & \multicolumn{2}{c|}{$x^*\geq C$} \\
[.05in]
\hline %%%%%%%%%%%%%%%%%%%%%%%%%%%%%%%%%%%%%%%%%%%%%%%%%%%
\hline %%%%%%%%%%%%%%%%%%%%%%%%%%%%%%%%%%%%%%%%%%%%%%%%%%%
\rule{0pt}{3ex}
\multirow{2}{*}{\begin{tabular}{l}$(C,\cdot)$\\offline \end{tabular}}&
$\rho\geq\frac{3}{2}\frac{\alpha-1+(2/3)^\alpha}{\alpha}$ &
%\cellcolor[gray]{.8} $\rho \geq\frac{13}{12}$ &
\cellcolor[gray]{.8} $\rho \geq\frac{11}{9}$ &
$\rho\geq\frac{3}{2}\frac{\alpha-1+(2/3)^\alpha}{\alpha}$ &
%\cellcolor[gray]{.8} $\rho\geq\frac{13}{12}$ \\
\cellcolor[gray]{.8} $\rho \geq\frac{11}{9}$ \\
[.05in]
\cline{2-5} %%%%%%%%%%%%%%%%%%%%%%%%%%%%%%%%%%%%%%%%%%%%%%%%%%%
\rule{0pt}{3ex}
&
$\rho<\frac{\OPTBP}{m^*}\left(1+\epsilon+\frac{1}{\alpha-1}+\frac{1}{\OPTBP}\right)$&
%\cellcolor[gray]{.8}$\rho<\frac{\OPTBP}{m^*}\left(2+\epsilon+\frac{1}{\OPTBP}\right)$&
\cellcolor[gray]{.8}$\rho<\frac{\OPTBP}{m^*}\left(\frac{3}{2}+\epsilon+\frac{1}{\OPTBP}\right)$&

$\rho<1+\epsilon+\frac{C^\alpha}{b}+\frac{1}{\OPTBP}$&
%\cellcolor[gray]{.8} $\rho<1+\epsilon+\frac{C^2}{b}+\frac{1}{\OPTBP}$\\
%\cellcolor[gray]{.8} $\rho<2+\epsilon+\frac{1}{\OPTBP}$\\
\cellcolor[gray]{.8} $\rho<\frac{3}{2}+\epsilon+\frac{1}{\OPTBP}$\\
[.05in]
\hline %%%%%%%%%%%%%%%%%%%%%%%%%%%%%%%%%%%%%%%%%%%%%%%%%%%
\hline %%%%%%%%%%%%%%%%%%%%%%%%%%%%%%%%%%%%%%%%%%%%%%%%%%%
\rule{0pt}{3ex}
\multirow{2}{*}{\begin{tabular}{l}$(C,\cdot)$\\online \end{tabular}}&
$\rho\geq\frac{(3/2)2^\alpha-1}{2^\alpha-1}$&
%\cellcolor[gray]{.8} $\rho\geq\frac{5}{3}$&
\cellcolor[gray]{.8} $\rho\geq\frac{11}{7}$&
$\rho\geq\frac{C^\alpha+2b}{b+\max\{C^\alpha,2(C/2)^\alpha+b\}}$&
%\cellcolor[gray]{.8} $\rho\geq\frac{C^2+2b}{b+\max\{C^2,C^2/2+b\}}$\\
%\cellcolor[gray]{.8} $\rho\geq\frac{6}{5}$\\
\cellcolor[gray]{.8} $\rho\geq\frac{20}{17}$\\
[.05in]
\cline{2-5} %%%%%%%%%%%%%%%%%%%%%%%%%%%%%%%%%%%%%%%%%%%%%%%%%%%
\rule{0pt}{3ex}
&
\begin{tabular}{c}\rule{0pt}{3ex}$\rho=1$ if $D_s=\emptyset$, else\\$\rho\leq\left(1-\frac{1}{\alpha}\left(1-\frac{1}{2^\alpha}\right)\right)\left(2+\frac{x^*}{\ell(D_s)}\right)$\\[.05in]\end{tabular}&
%\cellcolor[gray]{.8} $\rho\leq\frac{5}{4}+\frac{5x^*}{8\ell(D_s)}$&
%\cellcolor[gray]{.8} $\rho\leq\frac{5}{4}\left(1+\frac{1}{2\ell(D_s)}\right)$&
\cellcolor[gray]{.8} $\rho\leq\frac{17}{12}\left(1+\frac{1}{2\ell(D_s)}\right)$&
$\rho\leq\frac{2b}{C}\left( 1 + \frac{1}{(\alpha-1)2^{\alpha}} \right)\left(2 + \frac{C}{\ell(D)}\right)$&
%\cellcolor[gray]{.8} $\rho\leq\frac{5b}{C}+\frac{5b}{2\ell(D)}$\\
%\cellcolor[gray]{.8} $\rho\leq 5\left(1+\frac{1}{2\ell(D)}\right)$\\
\cellcolor[gray]{.8} $\rho\leq \frac{17}{2}\left(1+\frac{1}{2\ell(D)}\right)$\\
[.05in]
\hline %%%%%%%%%%%%%%%%%%%%%%%%%%%%%%%%%%%%%%%%%%%%%%%%%%%
\hline %%%%%%%%%%%%%%%%%%%%%%%%%%%%%%%%%%%%%%%%%%%%%%%%%%%
\rule{0pt}{3ex}
\multirow{2}{*}{\begin{tabular}{l}$(C,m)$\\online \end{tabular}}&
$\rho\geq\frac{(3/2)2^\alpha-1}{2^\alpha-1}$&
%\cellcolor[gray]{.8} $\rho\geq\frac{5}{3}$&
\cellcolor[gray]{.8} $\rho\geq\frac{11}{7}$&
$\rho\geq\frac{C^\alpha+2b}{b+\max\{C^\alpha,2(C/2)^\alpha+b\}}$&
%\cellcolor[gray]{.8} $\rho\geq\frac{6}{5}$\\
\cellcolor[gray]{.8} $\rho\geq\frac{20}{17}$\\
[.2in]
%\cline{2-5} %%%%%%%%%%%%%%%%%%%%%%%%%%%%%%%%%%%%%%%%%%%%%%%%%%%
%\rule{0pt}{3ex}
%&
%&
%\cellcolor[gray]{.8} &
%&
%\cellcolor[gray]{.8} \\
%[.05in]
\hline %%%%%%%%%%%%%%%%%%%%%%%%%%%%%%%%%%%%%%%%%%%%%%%%%%%
\hline %%%%%%%%%%%%%%%%%%%%%%%%%%%%%%%%%%%%%%%%%%%%%%%%%%%
\rule{0pt}{3ex}
\multirow{2}{*}{\begin{tabular}{l}$(\cdot,\cdot)$\\online \end{tabular}}&
$\rho\geq\frac{(3/2)2^\alpha-1}{2^\alpha-1}$&
%\cellcolor[gray]{.8} $\rho\geq\frac{5}{3}$&
\cellcolor[gray]{.8} $\rho\geq\frac{11}{7}$&
\multicolumn{2}{c|}{\multirow{2}{*}{not applicable}}\\
%\cellcolor[gray]{.8} \\
[.05in]
\cline{2-3} %%%%%%%%%%%%%%%%%%%%%%%%%%%%%%%%%%%%%%%%%%%%%%%%%%%
\rule{0pt}{3ex}
&
\begin{tabular}{c}\rule{0pt}{3ex}$\rho=1$ if $D_s=\emptyset$, else\\$\rho\leq\left(1-\frac{1}{\alpha}\left(1-\frac{1}{2^\alpha}\right)\right)\left(2+\frac{x^*}{\ell(D_s)}\right)$\\[.05in]\end{tabular}&
%\cellcolor[gray]{.8} $\rho\leq\frac{5}{4}+\frac{5x^*}{8\ell(D_s)}$&
%\cellcolor[gray]{.8} $\rho\leq\frac{5}{4}\left(1+\frac{1}{2\ell(D_s)}\right)$&
\cellcolor[gray]{.8} $\rho\leq\frac{17}{12}\left(1+\frac{1}{2\ell(D_s)}\right)$&
%$\rho<\frac{2b}{C}\left( 1 + \frac{1}{(\alpha-1)2^{\alpha}} \right)\left(2 + \frac{C}{\ell(D)}\right)$&
%\cellcolor[gray]{.8} $\rho<\frac{5b}{C}+\frac{5b}{2\ell(D)}$\\
&
%\cellcolor[gray]{.8} \\
\\
[.05in]
\hline %%%%%%%%%%%%%%%%%%%%%%%%%%%%%%%%%%%%%%%%%%%%%%%%%%%
\hline %%%%%%%%%%%%%%%%%%%%%%%%%%%%%%%%%%%%%%%%%%%%%%%%%%%
\rule{0pt}{3ex}
\begin{tabular}{l}$(\cdot,m)$\\online \end{tabular}&
$\rho\geq\max\{\frac{(3/2)2^\alpha-1}{2^\alpha-1},\frac{3^\alpha}{2^{\alpha+2}+\epsilon}\}$&
%\cellcolor[gray]{.8} $\rho\geq \frac{5}{3}$ &
\cellcolor[gray]{.8} $\rho\geq\frac{11}{7}$&
\multicolumn{2}{c|}{not applicable}\\
%\cellcolor[gray]{.8} \\
[.05in]
\cline{2-3} %%%%%%%%%%%%%%%%%%%%%%%%%%%%%%%%%%%%%%%%%%%%%%%%%%%
\rule{0pt}{3ex}
&
$\rho\leq O(\alpha)^\alpha$ In~\cite{PruhsPrimalDual}&
\cellcolor[gray]{.8}&
&
\\
%\cellcolor[gray]{.8} \\
[.05in]
\hline %%%%%%%%%%%%%%%%%%%%%%%%%%%%%%%%%%%%%%%%%%%%%%%%%%%
\hline %%%%%%%%%%%%%%%%%%%%%%%%%%%%%%%%%%%%%%%%%%%%%%%%%%%
\rule{0pt}{3ex}
\multirow{2}{*}{\begin{tabular}{l}\rule{0pt}{3ex}$(\cdot,2)$\\online\\[.05in] \end{tabular}}&
$\rho\geq\max\{\frac{3^\alpha}{2^{\alpha+1}},\frac{(3/2)2^\alpha-1}{2^\alpha-1},\frac{3^\alpha}{2^{\alpha+2}+\epsilon}\}$&
%\cellcolor[gray]{.8} $\rho\geq\frac{5}{3}$&
\cellcolor[gray]{.8} $\rho\geq\frac{11}{7}$&
\multicolumn{2}{c|}{\multirow{2}{*}{not applicable}}\\
%\cellcolor[gray]{.8} \\
[.05in]
\cline{2-3} %%%%%%%%%%%%%%%%%%%%%%%%%%%%%%%%%%%%%%%%%%%%%%%%%%%
\rule{0pt}{3ex}
&
\begin{tabular}{c}\rule{0pt}{3ex}$\rho=1$ if $\ell(D)\leq \sqrt[\alpha]{b/(2^\alpha-2)}$, else\\ $\rho\leq \max\{2,\left(\frac{3}{2}\right)^{\alpha-1}\}$\\[.05in]\end{tabular}&
%\cellcolor[gray]{.8} $\rho\leq 2$&
\cellcolor[gray]{.8} $\rho\leq \frac{9}{4}$&
&
\\
%\cellcolor[gray]{.8} \\
[.05in]
\hline %%%%%%%%%%%%%%%%%%%%%%%%%%%%%%%%%%%%%%%%%%%%%%%%%%%
\end{tabular}
%\caption{Summary of bounds on the approximation/competitive ratio $\rho$. 
%All lower bounds are existential. 
%The number of PMs in an optimal \CVMA solution is denoted as $m^*$. 
%The number of PMs in an optimal Bin Packing solution is denoted as \OPTBP. 
%The load that minimizes the ratio power consumption against load is denoted as $x^*$. The subset of VMs with load smaller than $x^*$ is denoted as $D_s$.
%Shaded cells correspond to $\alpha=2$, $b=1$, and $C=2$ on the left and $C=1$ on the right.
%Shaded cells correspond to $\alpha=3$, $b=2$, and $C=2$ on the left and $C=1$ on the right.
%}
\caption{Summary of bounds on the approximation/competitive ratio $\rho$. All lower bounds are existential. The number of PMs in an optimal \CVMA solution is denoted as $m^*$. The number of PMs in an optimal Bin Packing solution is denoted as \OPTBP. The load that minimizes the ratio power consumption against load is denoted as $x^*$. The subset of VMs with load smaller than $x^*$ is denoted as $D_s$. Shaded cells correspond to $\alpha=3$, $b=2$, and $C=2$ on the left and $C=1$ on the right.}

\label{table:results}
\end{table}
\renewcommand{\tabcolsep}{6pt}

In this work, we study offline and online versions of the four versions of the VMA problem.
For the offline problems, the first fact we observe is that there is a hard decision version of \CMVMA: Is there a feasible partition $\pi$ of the set $D$ of VMs? By reduction from the 3-Partition problem, it can be shown that this decision
problem is strongly NP-complete.

We then show that the \VMA, \CVMA, and \MVMA problems are NP-hard in the strong sense, 
even if $\alpha$ is constant. This result implies that these problems do not have FPTAS, even if $\alpha$ is constant.
However, we show that the \VMA and \MVMA problems have PTAS, while the \CVMA problem can not
be approximated beyond a ratio of $\frac{3}{2} \cdot \frac{\alpha-1 + (\frac{2}{3})^{\alpha}}{\alpha}$ (unless $\mathrm{P}=\mathrm{NP}$). 
On the positive side, we show how to use an existing Asymptotic PTAS~\cite{binpackingapprox} to obtain
algorithms that approximate the optimal solution of \CVMA. (See Table~\ref{table:results}.)

Then we move on to online VMA algorithms. We show various upper and lower bounds on the competitive ratio of the four versions of the problem.
(See Table~\ref{table:results}.) Observe that the results are often different depending on whether $x^*$ is smaller than $C$ or not. In fact, 
when $x^* < C$, there is a lower bound of $\frac{(3/2)2^\alpha-1}{2^\alpha-1}$ that applies to all versions of the problem.
The bounds are given as a function of the input parameters of the problem, in order to allow for tighter expressions.
%To provide intuition on how tight the bounds are, we instantiate them for a realistic value of $\alpha=2$, 
%and normalized values of $b=1$ and $C \in \{1,2\}$. The resulting bounds are shown in Table~\ref{table:results}
\jordi{To provide intuition on how tight the bounds are, we instantiate them for a realistic
\footnote{The values for $\alpha$ in the servers studied in \cite{eenergy} (denoted as Erdos and Nemesis) are close to $1.5$ and $3$ and $x^*$ values of $0.76C$ and $0.9C$ respectively.}
 value of $\alpha=3$, 
and normalized values of $b=2$ and $C \in \{1,2\}$. The resulting bounds are shown in Table~\ref{table:results}}
in shaded cells. As can be observed, the resulting upper and lower bounds are not very far in general.

\section{Preliminaries}
\label{sec:prelim}

The following claims will be used in the analysis.
%
%\begin{definition}
We call \defn{power rate} the power consumed per unit of load in a PM.
%\end{definition}
Let $x$ be the load of a PM. Then, its power rate is computed as $f(x)/x$. The load at which the power rate is minimized, denoted $x^*$, is the \defn{optimal load}, and the corresponding rate is the \defn{optimal power rate} $\varphi^*=f(x^*)/x^*$. Using calculus we get the following observation.

\begin{observation}
\label{opt-load}
The  optimal load is $x^* = \left(b/(\alpha-1)\right)^{1/\alpha}.$ Additionaly, for any $x\neq x^*$, $f(x)/x > \varphi^*$.
\end{observation}

%\begin{lemma}
%\label{l-all-xstar}
%Given an instance of the \EVMA problem with a set of VMs $D=\{d_1, \ldots, d_n\}$, any solution $\P=\{A_1, \ldots, A_m\}$ where $\sum_{d\in A_i}d\neq x^*$ for some
%$i\in[1,m]$, satisfies
%$$
%P(\P) > \rho^* \l(D) = \rho^* \sum_{d \in D} \l(d).
%$$
%\end{lemma}
%
%\begin{proof}
%The total cost of $\P$ is $\sum_{i\in[1,m]}f(\l(A_i))$ which, from Observation \ref{opt-load}, is at least
%\begin{eqnarray*}
%\sum_{i\in[1,m]:A_i\neq\emptyset} \l(A_i) \rho^*= \rho^* \sum_{i\in[1,m]:A_i\neq\emptyset} \sum_{d\in A_i} \l(d)= \rho^* \sum_{d\in D} \l(d).
%\end{eqnarray*}
%\end{proof}
%
%
%\begin{corollary}
%Given an instance of the \EVMA problem with VMs $D=\{d_1, \ldots, d_n\}$, any solution $\P$ satisfies
%$
%P(\P) \geq \rho^* \sum_{i=1}^n \l(d_i).
%$
%\end{corollary}

The following lemmas will be used in the analysis.
\begin{lemma}
\label{lem:xsmallerC}

Consider two solutions $\P=\{A_1, \ldots, A_m\}$ and $\P'=\{A'_1, \ldots, A'_m\}$ of an instance of the VMA problem, such that for some
$x,y \in [1,m]$ it holds that
\begin{itemize}
\item
$A_x\neq\emptyset$ and  $A_y\neq\emptyset$;
\item
$A'_x = A_x \cup A_y$,  $A'_y=\emptyset$, and $A_i=A'_i$, for all $i \neq x$ and $i \neq y$;
and
\item
$\l(A_x) + \l(A_y) \leq \min\{x^*,C\}$.
\end{itemize}
Then, $P(\P') < P(\P)$.
\end{lemma}

%\begin{comment}
\begin{proof}
Let $\ell(A_i)=x$ and $\ell(A_j)=y$.
First we notice that $\P'$ is feasible because $x+y\leq C$.
Now, using that $x+y \leq x^*$, we have
\begin{eqnarray*}
b	&=&  (x^*)^{\alpha}(\alpha - 1) \geq  (x+y)^{\alpha}(\alpha - 1) \\
  &>& (x+y)^{\alpha} \geq  (x+y)^{\alpha} -  (x^{\alpha}+y^{\alpha})
\end{eqnarray*}
where the second inequality comes from the fact that $\alpha > 1$. The above inequality is equivalent to
\begin{eqnarray*}
2b + x^\alpha+  y^\alpha > b+ (x+y)^\alpha,
\end{eqnarray*}
which implies the lemma.
\qed
\end{proof}
%\end{comment}

From this lemma, it follows that the global power consumption can be reduced by having $2$ VMs together in the same PM, when its aggregated load is smaller than $\min\{x^*,C\}$, instead of moving one VM to an unused PM. When we keep VMs together in a given partition we say that we are \emph{using} Lemma~\ref{lem:xsmallerC}.

%\begin{observation}
%\label{obs:pileupVMA}
%Lemma \ref{lem:xsmallerC} can be applied to \VMA exchanging clause $\l(A_x) + \l(A_y) \leq \min\{x^*,C\}$ by $\l(A_x) + \l(A_y) \leq x^*$
%\end{observation}

\begin{lemma}
\label{l-balance}
Consider two solutions $\P=\{A_1, \ldots, A_m\}$ and $\P'=\{A'_1, \ldots, A'_m\}$ of an instance of the VMA problem, such that for some
$x,y \in [1,m]$ it holds that
\begin{itemize}
\item
$A_x \cup A_y = A'_x \cup A'_y$, while $A_i=A'_i$, for all $x \neq i \neq y$;
\item
none of $A_x$,  $A_y$, $A'_x$, and $A'_y$ is empty;
and
\item
$|\l(A_x) - \l(A_y)| < |\l(A'_x) - \l(A'_y)|$.
\end{itemize}
Then, $P(\P) < P(\P')$.
\end{lemma}
%\begin{comment}
\begin{proof}
From the definition of $P(\cdot)$, to prove the claim is it enough to prove that $\l(A_x)^\alpha + \l(A_y)^\alpha < \l(A'_x)^\alpha + \l(A'_y)^\alpha$. Let us
assume wlog that $\l(A_x) \leq \l(A_y)$ and $\l(A'_x) \leq \l(A'_y)$. Let us denote $L=\l(A_x) + \l(A_y) = \l(A'_x) + \l(A'_y)$, and assume that
$\l(A_x) = \delta_1 L$ and $\l(A'_x) = \delta_2 L$. Note that $\delta_2 < \delta_1 \leq 1/2$. Then, the claim to be proven becomes
\begin{eqnarray*}
&(\delta_1 L)^\alpha + ((1-\delta_1)L)^\alpha < (\delta_2 L)^\alpha + ((1-\delta_2) L)^\alpha\\
&\delta_1^\alpha + (1-\delta_1)^\alpha < \delta_2^\alpha + (1-\delta_2)^\alpha
\end{eqnarray*}
Which holds because the function $f(x)=x^\alpha + (1-x)^\alpha$ is decreasing in the interval $(0,1/2)$.
\qed
\end{proof}

This lemma carries the intuition that balancing the load among the used PMs as much as possible reduces the power consumption.
%\end{comment}

\begin{corollary}
\label{cor:balance}
Consider a solution $\P=\{A_1, \ldots, A_m\}$ of an instance of the VMA problem with total load $\L$, such that exactly $k$ of the $A_x$ sets, $x \in [1,m]$, are non-empty (hence it uses $k$ PMs). Then, the power consumption is lower bounded by the power of the (maybe unfeasible) solution that balances the load evenly, i.e.,
$$
P(\P) \geq kb + k (\L/k)^\alpha.
$$
\end{corollary} 

\section{Offline Analysis}
\label{sec:alloffline}
%!TEX root = ./main.tex

\subsection{NP-hardness}

As was mentioned, it can be shown that deciding whether there is a feasible solution for an instance of the \CMVMA problem is NP-complete or not, by a direct
reduction from the 3-Partition problem. However, this result does not apply directly to the \CVMA, \MVMA, and \VMA problems. We show now that these
problems are NP-hard. We first prove the following lemma.

\begin{lemma}
\label{l-all-xstar}
Given an instance of the VMA problem, any solution $\P=\{A_1, \ldots, A_m\}$ where $\l(A_i) \neq x^*$ for some
$i\in[1,m] : A_i\neq\emptyset$, has power consumption
$
P(\P) > \rho^* \l(D) = \rho^* \sum_{d \in D} \l(d).
$
\end{lemma}

\begin{proof}
The total cost of $\P$ is $P(\P) = \sum_{i\in[1,m]}f(\l(A_i))$ which, from Observation \ref{opt-load}, satisfies
\begin{eqnarray*}
P(\P) &>& \sum_{i\in[1,m]:A_i\neq\emptyset} \l(A_i) \rho^*\\
		&=& \rho^* \sum_{i\in[1,m]:A_i\neq\emptyset} \sum_{d\in A_i} \l(d)= \rho^* \sum_{d\in D} \l(d).
\end{eqnarray*}
\qed
\end{proof}

%\begin{corollary}
%Given an instance of the \EVMA problem with VMs $D=\{d_1, \ldots, d_n\}$, any solution $\P$ satisfies
%$
%P(\P) \geq \rho^* \sum_{i=1}^n \l(d_i).
%$
%\end{corollary}

We show now in the following theorem that the different versions of the \CMVMA problem with unbounded $C$ or $m$ are NP-hard.

\begin{theorem}
The \CVMA, \MVMA and \VMA problems are strongly NP-hard, even if $\alpha$ is constant.
\end{theorem}
\begin{proof}
We show a reduction from  3-Partition defined as follows~\cite{gareyjohnson}, which is strongly NP-complete.
%from the strongly NP-complete problem 3-Partition defined as follows~\cite{gareyjohnson}.

INSTANCE:  Set $A$ of $3k$ elements, a bound $B\in\mathbb{Z}^+$ and, for each $a\in A$, a size $s(a)\in\mathbb{Z}^+$ such that $B/4<s(a)<B/2$ and $\sum_{a\in A} s(a) = kB$.

QUESTION: can $A$ be partitioned into $k$ disjoint sets $\{A_1,A_2,\dots,A_k\}$ such that $\sum_{a\in A_i} s(a) = B$  for each $1\leq i\leq k$?

The reduction is as follows. Given an instance of 3-Partition on a set $A=\{a_1,\dots,a_{3k}\}$ with bound $B$, and given a fixed value $\alpha>1$, we define an instance $\mathcal{I}$ of \VMA as follows:
$D = \{a_1,\dots,a_{3k}\}$, $\l(\cdot) = s(\cdot)$, and
$b = B^\alpha(\alpha-1)$ (i.e., $x^*=B$).
(For the proof of the \CVMA and \MVMA problems it is enough to set $C = B$ and $m = k$ when required.)
We show now that the answer to the 3-Partition problem is YES if and only if
the output $\P=\{A_1,A_2,\dots,A_m\}$ of the \VMA problem on input $\mathcal{I}$ is such that $\sum_{i=1}^m f(\l(A_i))=kf(B)$.

For the direct implication, assume that there exists a partition $\{A_1,A_2,\dots,A_k\}$ of $A$ such that
for each $i \in [1,k]$, $\sum_{a\in A_i} s(a) = B$. Then, in the context of the \VMA problem, such partition has cost $\sum_{i=1}^m f(\l(A_i))=kf(B)$. We claim that any partition has at least cost $kf(B)$.
In order to prove it, assume for the sake of contradiction that there is a partition $\P'=\{A_1',A_2',\dots,A_m'\}$ of \VMA on input $\mathcal{I}$ with cost less than $kf(B)$. Then, there is some $i\in[1,m]$ such that $A'_i \neq \emptyset$ and $\l(A'_i)\neq B$. From Lemma~\ref{l-all-xstar}, $P(\P')> \rho^* \l(D) = (f(x^*)/x^*)kB$. Since $B=x^*$, we have that $P(\P') > k f(B)$, which is a contradiction.

To prove the reverse implication, assume an output $\P=\{A_1,A_2,\dots,A_m\}$ of the \VMA problem on input $\mathcal{I}$ such that $P(\P) =\sum_{i=1}^m f(\l(A_i))=kf(B)$. Then, it must be $\forall i\in[1,m]:A_i\neq\emptyset, \l(A_i)=B$. Otherwise, from Lemma~\ref{l-all-xstar}, $P(\P) > k f(B)$, a contradiction.
\qed
\end{proof}

It is known that strongly NP-hard problems cannot have a fully polynomial-time approximation scheme (FPTAS)~\cite{citeulike:556492}. Hence, the following corollary.

\begin{corollary}
The \CVMA, \MVMA and \VMA problems do not have fully polynomial-time approximation schemes (FPTAS), even if $\alpha$ is constant.
\end{corollary}

In the following sections we show that, while the \MVMA and \VMA problems have polynomial-time approximation schemes (PTAS), the \CVMA problem
cannot be approximated below $\frac{3}{2} \cdot \frac{\alpha-1 + (2/3)^{\alpha}}{\alpha}$.

%------------------------------------------------------------------------------------------------

\subsection{The \MVMA and \VMA Problems Have PTAS}
\label{sec:offmvma}

We have proved that the \MVMA and \VMA problems are NP-hard in the strong sense and that, hence, there exists no FPTAS for them.
However, Alon et al.~\cite{alon1998approximation}, proved that if a function $f(\cdot)$ satisfies a condition denoted $F*$, then
the problem of scheduling jobs in $m$ identical machines so that $\sum_i{f(M_i)}$ is minimized has a PTAS, where $M_i$ is the
load of the jobs allocated to machine $i$. This result implies that if our function $f(\cdot)$ satisfies condition $F*$, the same PTAS
can be used for the \MVMA and \VMA problems.
From Observation 6.1 in \cite{epstein2004}, it can be derived that, in fact, our power consumption function $f(\cdot)$ satisfies condition $F*$.
%\anto{(A direct proof that $f(\cdot)$ satisfies $F*$ can be found in ???.)}
Hence, the following theorem.

\begin{theorem}
There are polynomial-time approximation schemes (PTAS) for the \MVMA and \VMA problems.
\end{theorem}

%------------------------------------------------------------------------------------------------

\subsection{Bounds on the Approximability of the \CVMA Problem}
\label{sec:offcvma}

We study now the \CVMA problem, where we consider an unbounded number of machines with bounded capacity $C$. We will provide a lower bound on its approximation ratio, independently on the relation between $x^*$ and $C$; and upper bounds for the cases when $x^*\geq C$ and $x^*<C$.

\subsubsection{Lower bound on the approximation ratio}
%\textbf{Lower bound on the approximation ratio}

The following theorem shows a lower bound on the approximation ratio of any offline algorithm for \CVMA.

\begin{theorem}
\label{th:offlineLB}
No algorithm achieves an approximation ratio smaller than
$\frac{3}{2} \cdot \frac{\alpha-1 + (\frac{2}{3})^{\alpha}}{\alpha}$ for the \CVMA problem unless $\mathrm{P}=\mathrm{NP}$.
\end{theorem}

\begin{proof}
The claim is proved showing a reduction from the partition problem~\cite{gareyjohnson}. In the partition problem there is
a set $A=\{a_1, a_2, \ldots, a_n\}$ of $n$ elements, there is a size $s(a)$ for each element $a\in A$, and the sum $M = \sum_{a \in A} s(a)$ of the
sizes of the elements in $A$. The problem decides whether there is a subset $A' \subset A$ such that $\sum_{a \in A'} s(a) = M/2$.

From an instance of the partition problem, we construct an instance of the \CVMA problem as follows.  The set of VMs in the system is
$D=\{a_1,a_2,\ldots,a_n\}$, the load function is $\ell(\cdot)=s(\cdot)$, the capacity of each PM is set to $C=M/2$, and
$b$ is set to $b=C^\alpha (\alpha-1)$ (i.e., $x^*=C$). Let us study the optimal partition $\pi^*$
such that the total power consumption $P(\pi^*)$ is minimized.
If there is a partition of $D$ such that each subset in this partition has load $M/2$ then, from Observation~\ref{opt-load},
$\pi^*$ has all the VMs assigned to two PMs.
Otherwise, $\pi^*$ needs at least $3$ PMs to allocate all the VMs. From Corollary~\ref{cor:balance}, the power consumption of this solution is lower bounded by the power of a (maybe unfeasible) partition that balances the load among the 3 PMs as evenly as possible. Formally,
\begin{eqnarray*}
	&\exists A': \sum_{a \in A'} s(a) = M/2 \\
	\Rightarrow  &P(\pi^*) = 2b+2\left(\frac{M}{2}\right)^\alpha = 2b+2 C^\alpha\\
	&\nexists A': \sum_{a \in A'} s(a) = M/2 \\
	\Rightarrow  &P(\pi^*) \geq 3b+3\left(\frac{M}{3}\right)^\alpha = 3b+3\left(\frac{2C}{3}\right)^\alpha.
\end{eqnarray*}
Comparing both values we obtain the following ratio.
\begin{eqnarray*}
\rho&=&\frac{3b+3\left(\frac{2C}{3}\right)^\alpha}{2b+2C^\alpha} =
     \frac{3C^\alpha(\alpha-1)+3\left(\frac{2C}{3}\right)^\alpha}{2C^\alpha(\alpha-1)+2C^\alpha} \\
&=& \frac{3}{2} \cdot \frac{\alpha-1 + (\frac{2}{3})^{\alpha}}{\alpha}.
\end{eqnarray*}
Therefore, given any $\epsilon > 0$, having a polynomial-time algorithm ${\cal A}$ with approximation ratio $\rho- \epsilon$
would imply that this algorithm could be used to decide if there is a subset $A' \subset A$ such that $\sum_{a \in A'} s(a) = M/2$.
In other words, this algorithm would be able to solve the partition problem.
This contradicts the fact that the partition problem is NP-hard and no polynomial time algorithm solves it unless $\mathrm{P}=\mathrm{NP}$.
Therefore, there is no algorithm that achieves a $\rho-\epsilon=\frac{3}{2} \cdot \frac{\alpha-1 + (\frac{2}{3})^{\alpha}}{\alpha} - \epsilon$ approximation ratio for the \CVMA problem unless $\mathrm{P}=\mathrm{NP}$.
\qed
\end{proof}

\subsubsection{Upper bound on the approximation ratio for $x^*\geq C$}
\label{sec:UpboundXlargC}

We study now an upper bound on the competitive ratio of the \CVMA problem for the case when $x^*\geq C$. Under this condition, the best is to load
each PM to its full capacity. Intuitively, an optimal solution should load every machine up to its maximum capacity or, if not possible, should balance the load among PMs to maximize the average load.  The following lemma formalizes this observation.

\begin{lemma}
\label{lem:lbXlargerC}
For any system with unbounded number of PMs where $x^*\geq C$
% and the aggregated load of all VMs is $L$,
the power consumption of the optimal assignment $\pi^*$ is lower bounded by
the power consumption of a (possibly not feasible) solution where $\l(D)$ is evenly distributed among \OPTBP PMs,
where \OPTBP is the minimum number of PMs required to allocate all VMs (i.e., the optimal solution of the packing problem).
That is, $P(\pi^*) \geq \OPTBP \cdot b+ \OPTBP (\l(D)/\OPTBP)^\alpha$.
\end{lemma}

%\begin{comment}
\begin{proof}
Denote the number of PMs used in an optimal \CVMA solution $\pi^*$ by $m^*$. By Corollary~\ref{cor:balance}, we know that
$P(\pi^*)  \geq m^*b + m^*  (\l(D)/m^*)^{\alpha}$.
%	\begin{equation*}
%		P(\pi^*) \geq m^*b + m^*  (\l(D)/m^*)^{\alpha}.
%	\end{equation*}
Given that $\OPTBP \leq m^*$, we know that $\l(D)/m^* \leq \l(D)/\OPTBP \leq  C\leq x^*$. Thus, for evenly-balanced loads the power consumption is reduced if the number of PMs is reduced, that is
$m^* b + m^*  (\l(D)/m^*)^{\alpha} \geq \OPTBP \cdot b + \OPTBP  (\l(D)/\OPTBP)^{\alpha}$.
%	\begin{equation*}
%		m^* b + m^*  (\l(D)/m^*)^{\alpha} \geq \OPTBP \cdot b + \OPTBP  (\l(D)/\OPTBP)^{\alpha}.
%	\end{equation*}
Hence, the claim follows.
\qed
\end{proof}

Now we prove an upper bound on the approximation ratio showing a reduction to bin packing~\cite{gareyjohnson}. The reduction works as follows. Let each PM be seen as a bin of capacity $C$, and each VM be seen as an object to be placed in the bins, whose size is the VM load. Then, a solution for this bin packing problem instance yields a feasible (perhaps suboptimal) solution for the instance of \CVMA.
Moreover, using any bin-packing approximation algorithm, we obtain a feasible solution for \CVMA that approximates the minimal number of PMs used.
The power consumption of this solution approximates the power consumption of the optimal solution $\pi^*$ of the instance of \CVMA.
In order to compute an upper bound on the approximation ratio of this algorithm, we will compare the power consumption of such solution against a lower bound on the power consumption of $\pi^*$.
The following theorem shows the approximation ratio obtained.

\begin{theorem}
\label{theo:ubxlargercBP}
For every $\epsilon >0$, there exists an approximation algorithm for the \CVMA problem when $x^*\geq C$ that achieves an approximation ratio of
\begin{equation*}
\rho < 1+\epsilon + \frac{ C^\alpha}{b} +  \frac{1}{\OPTBP},
\end{equation*}
%for any system with unbounded number of PMs where $x^*\geq C$,
where $\OPTBP$ is the minimum number of PMs required to allocate all the VMs.
\end{theorem}

%\begin{comment}
\begin{proof}
Consider an instance of the \CVMA problem. 
%Recall where the aggregated load of the VMs is $\l(D)$.
If $\l(D) \leq C$, the optimal solution is to place all the VMs in one single PM. Hence, we assume in the rest
of the proof that $\l(D)>C$.
Define the corresponding instance of bin packing following the reduction described above.
Let the optimal number of bins to accommodate all VMs be $\OPTBP$.
As shown in~\cite{binpackingapprox}, for every $\epsilon>0$, there is a polynomial-time algorithm that fits all VMs in $\widehat{m}$ bins,
where $\widehat{m} \leq (1+\epsilon)\OPTBP+1$.
From Lemma~\ref{l-balance}, once the number of PMs used $\widehat{m}$ is fixed, the power consumption
is maximized when the load is unbalanced to the maximum.
I.e., the power consumption of the assignment is at most $\widehat{m}b+(\l(D)/C) C^\alpha$.
On the other hand, as shown in Lemma~\ref{lem:lbXlargerC}, the power consumption of the optimal \CVMA solution is at least
$\OPTBP \cdot b+\OPTBP \left(\frac{\l(D)}{\OPTBP}\right)^\alpha$.
Then, we compute a bound on the approximation ratio as follows.
\begin{equation}
\label{eq:upboundXlargerC}
\rho \leq  \frac{\widehat{m}b+\left(\frac{\l(D)}{C}\right) C^\alpha}{\OPTBP \cdot b+\OPTBP \left(\frac{\l(D)}{\OPTBP}\right)^\alpha}
< \frac{\widehat{m}b+\left(\frac{\l(D)}{C}\right) C^\alpha}{\OPTBP \cdot b+\OPTBP \left(\frac{C}{2}\right)^\alpha},
\end{equation}
where the second inequality comes from $\l(D)/\OPTBP > C/2$. (If $\l(D)/\OPTBP \leq C/2$, there must be two PMs whose loads add up to less than $C$, which contradicts the fact that $\OPTBP$ is the number of bins used in the optimal solution of bin packing.)
Let $\gamma=(x^*/C)^\alpha$. Then, replacing $b=\gamma C^\alpha(\alpha-1)$, in Eq.~(\ref{eq:upboundXlargerC}) we have
\begin{eqnarray}
\rho
& < &  \frac{\widehat{m}\gamma C^\alpha(\alpha-1)+\left(\frac{\l(D)}{C}\right) C^\alpha}{\OPTBP \gamma C^\alpha(\alpha-1)+\OPTBP \left(\frac{C}{2}\right)^\alpha} \nonumber  \\ 
  &=&  \frac{\widehat{m}\gamma(\alpha-1)+\left(\frac{\l(D)}{C}\right)}{\OPTBP \gamma(\alpha-1)+\OPTBP \left(\frac{1}{2}\right)^\alpha} 
\leq   \frac{\widehat{m}\gamma(\alpha-1)+\OPTBP}{\OPTBP \gamma(\alpha-1)+\left(\frac{\OPTBP}{2^\alpha}\right)} \label{eq--9} \\
 & \leq &  \frac{(\OPTBP(1+\epsilon)+1))\gamma(\alpha-1)+\OPTBP}{\OPTBP\gamma(\alpha-1)+\left(\frac{\OPTBP}{2^\alpha}\right)} \label{eq--10} \\
 & = &  \frac{(1+\epsilon)\gamma(\alpha-1)+1}{\gamma(\alpha-1)+\left(\frac{1}{2^\alpha}\right)} +  \frac{\gamma(\alpha-1)}{\OPTBP\gamma(\alpha-1)+\left(\frac{\OPTBP}{2^\alpha}\right)} \nonumber \\
 & = & \frac{2^\alpha((1+\epsilon)\gamma(\alpha-1)+1)}{2^\alpha\gamma(\alpha-1)+1} +  \frac{2^\alpha\gamma(\alpha-1)}{\OPTBP ( 2^\alpha \gamma(\alpha-1)+1)} \nonumber \\
 & < &  \frac{(1+\epsilon)\gamma(\alpha-1)+1}{\gamma(\alpha-1)} +  \frac{1}{\OPTBP} \nonumber \\
 &=&   1+\epsilon + \frac{1}{\gamma(\alpha-1)} +  \frac{1}{\OPTBP} = 1+\epsilon + \frac{ C^\alpha}{b} +  \frac{1}{\OPTBP} \nonumber
\end{eqnarray}
%\textcolor{red}{There is a mistake in this proof. (extra exponent $\alpha$)}\\
Inequality (\ref{eq--9}) follows from $\l(D)/C \leq \OPTBP$, Inequality (\ref{eq--10}) from the approximation algorithm for bin packing,
and the last inequality is because $\OPTBP > 0$.
\qed
\end{proof}
%\end{comment}

\subsubsection{Upper bound on the approximation ratio for $x^* < C$}
%\label{ubxsmallercBP}

We study now the \CVMA problem when $x^* < C$. In this case, the optimal load per PM is less than its capacity, so an optimal solution would load every PM to $x^*$ if possible, or try to balance the load close to $x^*$. In this case we slightly modify the bin packing algorithm described above, reducing the bin size from $C$ to $x^*$. Then, using an approximation algorithm for this bin packing problem, the following theorem can be shown.

\begin{theorem}
\label{theo:ubxsmallercBP}
For every $\epsilon >0$, there exists an approximation algorithm for the \CVMA problem when $x^* < C$ that achieves an approximation ratio of
\begin{equation*}
\rho < \frac{\OPTBP}{m^*}\left((1+\epsilon)+\frac{1}{\alpha-1}\right) + \frac{1}{m^*},
\end{equation*}
%for any system with unbounded number of PMs where $x^*\geq C$,
where $m^*$ is the number of PMs used by the optimal solution of \CVMA, and $\OPTBP$ is the minimum number of PMs required to allocate all the VMs
without exceeding load $x^*$ (i.e., the optimal solution of the bin packing problem).
\end{theorem}

%\begin{comment}
\begin{proof}
Consider an instance of the \CVMA problem.
If $\l(D) \leq x^*$ then the optimal solution is to assign all the VMs to one
single PM. Then, in the rest of the proof we assume that $\l(D)>x^*$.
Assuming $m^*$ to be the number of PMs of an optimal \CVMA solution $\pi^*$ for load $\l(D)$, from Corollary \ref{cor:balance}, we can claim that the power consumption $P(\pi^*)$ can be bounded as $P(\pi^*) \geq m^*b+m^*(\l(D)/m^*)^\alpha$.

Now, let $\OPTBP$ be the minimum number of PMs required to allocate all the VMs of the \CVMA problem
without exceeding load $x^*$.
As shown in~\cite{binpackingapprox}, for every $\epsilon>0$, there is a polynomial-time algorithm that fits all VMs in $\widehat{m}$ bins,
where $\widehat{m} \leq (1+\epsilon)\OPTBP+1$.
From Lemma~\ref{l-balance}, this approximation results in a power consumption no larger than $\widehat{m}b+(\l(D)/x^*) (x^*)^\alpha$.
Hence, the approximation ratio $\rho$ of the solution obtained wit this algorithm can be bounded as follows.
\begin{equation}
\label{eq:upboundXsmallerC}
\rho \leq  \frac{\widehat{m}b+\left(\frac{\l(D)}{x^*}\right) (x^*)^\alpha}{m^*b+m^*\left(\frac{\l(D)}{m^*}\right)^\alpha}.
\end{equation}
Since $\l(D)>x^*$, we know that $\l(D)/m^* > x^*/2$, since otherwise there are two used PMs whose load is no larger than $x^*$, contradicting
by Lemma~\ref{lem:xsmallerC} the definition of $m^*$.
Also, from the definition of $\OPTBP$, it follows that $\l(D)\leq \OPTBP \cdot x^*$. Finally, recall that $b=(x^*)^\alpha(\alpha-1)$. Applying these results to Eq. (\ref{eq:upboundXsmallerC}) we have the following.
\begin{eqnarray*}
\rho & < &  \frac{\widehat{m} (x^*)^\alpha(\alpha-1)+\left(\frac{x^* \OPTBP}{x^*}\right) (x^*)^\alpha}{m^* (x^*)^\alpha(\alpha-1)+m^*\left(\frac{x^*}{2}\right)^\alpha} \\
  &=&  \frac{\widehat{m}(\alpha-1)+\OPTBP}{m^*(\alpha-1)+m^*\left(\frac{1}{2}\right)^\alpha} 
 \leq  \frac{(\OPTBP(1+\epsilon)+1)(\alpha-1)+\OPTBP}{m^*(\alpha-1)+\frac{m^*}{2^\alpha}}\\
 & = &  \frac{\OPTBP (1+\epsilon)(\alpha-1)+\OPTBP}{m^*(\alpha-1)+\frac{m^*}{2^\alpha}} +  \frac{\alpha-1}{m^*(\alpha-1)+\frac{m}{2^\alpha}}\\
 & = & \frac{\OPTBP}{m^*}\frac{2^\alpha((1+\epsilon)(\alpha-1)+1)}{2^\alpha(\alpha-1)+1} +  \frac{2^\alpha(\alpha-1)}{2^\alpha m^*(\alpha-1)+m^*}\\
 & \leq & \frac{\OPTBP}{m^*}\left((1+\epsilon)+\frac{1}{\alpha-1}\right) + \frac{1}{m^*},\\
\end{eqnarray*}
where the first inequality comes from applying the results aforementioned, and second one from using $\widehat{m}=\OPTBP(1+\epsilon)+1$, while the last one results from simplifying the previous equation.
\qed
\end{proof}
%\end{comment}

\section{Online Analysis}
\label{sec:comp}
In this section, we study the online version of the VMA problem, i.e.,  when the VMs are revealed one by one. We first study lower bounds and then provide online algorithms and prove upper bounds on their competitive ratio.
%We prove lower and upper bounds

	%\subsection{Competitiveness for \VMA}
	%\label{sec:onlinecvma}
    	%\subsubsection{Lower Bound for Unbounded Number of Physical Machines}
    	%\label{sec:lower-ub}
    	%\input{lower-unbounded}

	    %\subsubsection{Upper Bound for Unbounded Number of Physical Machines}
	    %\label{sec:UBunbound}
	    %\input{inftyPM}

	%\subsection{Competitiveness for \MVMA}
	%\label{sec:onlinemvma}

		%\subsubsection{Lower Bound for $m$ Physical Machines}
    	%\label{sec:lowerm}
		%\input{lower-mPM-3.tex}

		%\subsubsection{Upper Bound for $m$ Physical Machines}
    	%\label{sec:uppererm}
		%Not yet :)
		%\subsubsection{A Stronger Lower Bound for $2$ Physical Machines}
    	%\label{sec:lower2}
    	%\input{lower-nequalm}

	    %\subsubsection{Upper Bound for $2$ Physical Machines}
    	%\label{sec:upper2}
   		%\input{upp_2servers}

	\subsection{Lower Bounds}
	\label{sec:online-lower}
	%!TEX root = main.tex

In this section, we compute lower bounds on the competitive ratio for \VMA, \CVMA, \MVMA, \CMVMA and $(\cdot,2)$-VMA problems. We start with one general construction that is used to obtain
lower bounds on the first four cases. Then, we develop special constructions for \MVMA and  $(\cdot,2)$-VMA that improve the lower bounds for these two problems.

\subsubsection{General Construction}
We prove lower bounds on the competitive ratio of \VMA, \CVMA, \MVMA and \CMVMA problems. These lower bounds are shown in the following two theorems. In Theorem~\ref{thm:online-lower-vma-cvma-bigc}, we prove a
lower bound on the competitive ratio that is valid in the cases when $C$ is unbounded and when it is larger or equal than $x^*$. The case $C \leq x^*$ is covered in Theorem \ref{thm:online-lower-cvma-smallc}.

\begin{theorem}
\label{thm:online-lower-vma-cvma-bigc}
	There exists an instance of problems \VMA, \MVMA, \CVMA and \CMVMA when $C > x^*$, such that no online algorithm can guarantee a competitive ratio smaller than $\frac{(3/2)2^\alpha-1}{2^\alpha-1}$.
	% is no online algorithm that achieves a competitive ratio better than $(3-2^{1-\alpha})/(2-2^{1-\alpha})$ for  \VMA, \MVMA, and for \CVMA and \CMVMA when $C \geq x^*$.
\end{theorem}

\begin{proof}
We consider a scenario where, for any online algorithm, an adversary injects VMs of size $\epsilon x^*$ ($\epsilon>0$ is an arbitrarily small constant) to the system until the algorithm starts up a new PM. Let us assume that the total number of VMs injected is $k$. According to the adversary's behavior, the assignment of the VMs should be that all the VMs except one are allocated to a single PM while the second PM has only one VM. Depending on what the optimal solution is, we discuss the following two cases: \\
%\begin{itemize}
%\item {
\emph{Case 1:} $k \leq \frac{1}{\epsilon}\left(\frac{\alpha-1}{1-2^{1-\alpha}}\right)^{1/\alpha}$. The optimal solution will allocate all the VMs to a single PM. Consequently, the competitive ratio of the online algorithm satisfies
\begin{equation*}
	\rho(k) \geq \lim_{\epsilon \rightarrow 0}\left(\frac{\left((k-1)\epsilon x^*\right)^{\alpha} +(\epsilon x^*)^{\alpha} + 2b}{\left(k\epsilon x^*\right)^{\alpha}+b}\right).
\end{equation*}
It can be easily verified that function $\rho(k)$ is monotone decreasing with $k$. That is, $\rho(k)$ is minimized when $k=\frac{1}{\epsilon}\left(\frac{\alpha-1}{1-2^{1-\alpha}}\right)^{1/\alpha}$. As a result, we obtain,
\begin{align*}
	\rho(k) &\geq  \lim_{\epsilon \rightarrow 0}\left(\frac{\left(\left(\frac{\alpha-1}{1-2^{1-\alpha}}\right)^{1/\alpha}x^*\right)^{\alpha} + (\epsilon x^*)^{\alpha}+ 2b}{\left(\left(\frac{\alpha-1}{1-2^{1-\alpha}}\right)^{1/\alpha}x^*\right)^{\alpha}+b}\right) \\
	&= \frac{\left(\left(\frac{\alpha-1}{1-2^{1-\alpha}}\right)^{1/\alpha}x^*\right)^{\alpha} + 2(x^*)^{\alpha}(\alpha-1)}{\left(\left(\frac{\alpha-1}{1-2^{1-\alpha}}\right)^{1/\alpha}x^*\right)^{\alpha}+(x^*)^{\alpha}(\alpha-1)} \\
	&= \frac{3-2^{1-\alpha}}{2-2^{1-\alpha}}=\frac{(3/2)2^\alpha-1}{2^\alpha-1}.
	%\overset{\alpha \rightarrow \infty}{\Longrightarrow}\frac{3}{2}.
\end{align*}
%\item{
\emph{Case 2:} $k > \frac{1}{\epsilon}\left(\frac{\alpha-1}{1-2^{1-\alpha}}\right)^{1/\alpha}$. The optimal solution will use two PMs with $k/2$ PMs assigned to each PM. Accordingly, the competitive ratio of the online algorithm satisfies
\begin{equation*}
	\rho(k) \geq \lim_{\epsilon \rightarrow 0}\left(\frac{\left((k-1)\epsilon x^*\right)^{\alpha} + (\epsilon x^*)^{\alpha} + 2b}{2\left(\frac{k\epsilon x^*}{2}\right)^{\alpha}+2b}\right).
\end{equation*}
Similarly, we observe that $\rho(k)$ is monotone increasing with $k$. Consequently, the following inequality applies.
\begin{align*}
	\rho(k) & \geq  \lim_{\epsilon \rightarrow 0}\left(\frac{\left(\left(\frac{\alpha-1}{1-2^{1-\alpha}}\right)^{1/\alpha}x^*\right)^{\alpha} + (\epsilon x^*)^{\alpha}+ 2b}{2\left(\frac{1}{2}\left(\frac{\alpha-1}{1-2^{1-\alpha}}\right)^{1/\alpha}x^*\right)^{\alpha}+2b}\right) \\
	&= \frac{\left(\left(\frac{\alpha-1}{1-2^{1-\alpha}}\right)^{1/\alpha}x^*\right)^{\alpha} + 2(x^*)^{\alpha}(\alpha-1)}{2\left(\frac{1}{2}\left(\frac{\alpha-1}{1-2^{1-\alpha}}\right)^{1/\alpha}x^*\right)^{\alpha}+ 2(x^*)^{\alpha}(\alpha-1)} \\
	&= \frac{3-2^{1-\alpha}}{2-2^{1-\alpha}}
	=\frac{(3/2)2^\alpha-1}{2^\alpha-1}
	%\overset{\alpha \rightarrow \infty}{\Longrightarrow}\frac{3}{2}.
\end{align*}
Note that it can also happen that $C < \left(\frac{\alpha-1}{1-2^{1-\alpha}}\right)^{1/\alpha} x^*$. In this case, $k$ is smaller than $\frac{1}{\epsilon}\left(\frac{\alpha-1}{1-2^{1-\alpha}}\right)^{1/\alpha}$. Therefore, the competitive ratio
 %we can achieve
 is always larger than $\frac{(3/2)2^\alpha-1}{2^\alpha-1}$, proving the lower bound.
 \qed
%\end{itemize}
\end{proof}

\begin{theorem}
\label{thm:online-lower-cvma-smallc}
There exists an instance of problems \CVMA and \CMVMA  when $C \leq x^*$ such that no online algorithm can guarantee a competitive ratio smaller than $(C^\alpha + 2b)/(b+\max(C^\alpha,2(C/2)^\alpha+b))$.
%
%	 There is no online algorithm that achieves a competitive ratio better than $(C^\alpha + 2b)/(b+\max(C^\alpha,2(C/2)^\alpha+b))$ for \CVMA and \CMVMA  when $C < x^*$.
\end{theorem}

\begin{proof}
Similarly to the proof of Theorem~\ref{thm:online-lower-vma-cvma-bigc}, we prove the result by considering an adversarial injection of VMs of size $\epsilon C$. This injection stops when a new PM started up by an online algorithm. We discuss the following two cases: \\
%\begin{itemize}
%\item{
\emph{Case 1:} $k \leq 1/\epsilon $. In this case, the optimal algorithm will assign all the VMs to a single PM. The competitive ratio of the online algorithm satisfies
\begin{align*}
	\rho(k)  & \geq \lim_{\epsilon \rightarrow 0} \frac{\left((k-1)\epsilon C\right)^{\alpha} + (\epsilon C)^{\alpha} + 2b}{(k \epsilon C)^\alpha + b} \\
		& \geq \lim_{\epsilon \rightarrow 0} \frac{(1-\epsilon)^{\alpha} C^{\alpha} + 2b}{C^{\alpha} + b} \\%\label{eq:online-lower-cvma-smallc-2}\\
		& \geq \frac{C^{\alpha} + 2b}{C^{\alpha} + b} \geq 2 - \frac{1}{\alpha} %\label{eq:online-lower-cvma-smallc-3}.
\end{align*}
%Inequality (\ref{eq:online-lower-cvma-smallc-2})
The second inequality results from applying $k \leq 1/\epsilon$, which is observed from the monotone decreasing property of function $\rho(k)$. %Inequalities in Eq. (\ref{eq:online-lower-cvma-smallc-3})
The last inequality comes from computing the limit when $\epsilon$ goes to $0$ and by applying $b \geq C^{\alpha}(\alpha - 1)$. \\
%\item{
\emph{Case 2:} $k > 1/\epsilon$. In this case, the adversary stops injecting VMs as there will be, mandatorily, two active PMs, one of them not capable to allocate more VMs and the second one hosting one single VM. Since all the VMs can not be consolidated to a single PM. The optimal solution would use also two PMs but evenly balancing the loads among them. The competitive ratio of the online algorithm satisfies
\begin{align*}
	\rho(k) & = \lim_{\epsilon \rightarrow 0}\left(\frac{\left((k-1)\epsilon C\right)^{\alpha} + (\epsilon C)^{\alpha} + 2b}{2\left(\frac{k\epsilon C}{2}\right)^{\alpha}+2b}\right) \\
	        & = \lim_{\epsilon \rightarrow 0}\left(\frac{C^\alpha + (\epsilon C)^{\alpha} + 2b}{2\left(\frac{C+\epsilon C}{2}\right)^{\alpha}+2b}\right)
	         = \frac{C^\alpha + 2b}{2\left(\frac{C}{2}\right)^{\alpha}+2b}.
\end{align*}
Hence, combining the results from both cases $1$ and $2$ we obtain the bound presented in Theorem \ref{thm:online-lower-cvma-smallc}.
\qed
%\end{itemize}
\end{proof}

\subsubsection{Special Constructions for \MVMA and  $(\cdot,2)$-VMA} We show first that for $m$ PMs %($m$ is a bounded constant)
there is a lower bound on the competitive ratio that improves the previous lower bound when $\alpha > 4.5$. Secondly, we prove a particular lower bound for problem $(\cdot,2)$-VMA, that improves the previous lower bound when $\alpha > 3$.
\begin{theorem}
\label{thm:lbmservers}
There exists an instance of problem \MVMA such that no online algorithm can guarantee a competitive ratio smaller than $3^\alpha/(2^{\alpha+2} + \epsilon)$ for any $\epsilon > 0$.
%There is no online algorithm for \MVMA that achieves a competitive ratio better than $3^\alpha/(2^{\alpha+2} + \epsilon)$ for any $\epsilon > 0$.
\end{theorem}

%\begin{comment}
\begin{proof}
We prove the result by giving an adversarial arrival of VMs.
We evaluate the competitive ratio of any online algorithm \alg with respect to an algorithm \opt that distributes the VMs among all the PMs ``as evenly as possible''.
We define a value $\beta > 1$ such that $\epsilon \geq (\alpha - 1)/\beta^\alpha$ for some value $\epsilon > 0$. Note that such value $\beta$ can be defined for any $\epsilon > 0$.
The adversarial arrival follows.
In a first phase, $m$ VMs arrive, each with load $\beta x^*$.

Let $\pi$ be the partition given by \alg.
We show first that if $\pi$ uses less than $3m/4$ PMs\footnote{For clarity we omit floors and ceilings in the proof.}
or some PM is assigned more than 2 VMs there exists another partition that can be obtained from $\pi$, it uses exactly $3m/4$ PMs, no PM is assigned more than 2 VMs, and the power consumption is not worse.

If $\pi$ uses less than $3m/4$ PMs,
then there exists another partition $\pi'$ that uses exactly $3m/4$ PMs with a power consumption that is not worse than $P(\pi)$.
To see why, notice that there are PMs in $\pi$ that are assigned more than one VM and that each load is $\beta x^*>x^*$. Then, applying repeatedly 
Lemma~\ref{lem:xsmallerC}
%Observation~\ref{obs:pileupVMA} 
until $3m/4$ PMs are used, where $\l_1$ and $\l_2$ are the loads of any pair of VMs assigned to the same PM, a partition $\pi'$ such that $P(\pi') \leq P(\pi)$ can be obtained.

If in $\pi'$ some PM is assigned more than 2 VMs,
then there exists another partition $\pi''$ where no PM is assigned more than 2 VMs with a power consumption that is not worse than $P(\pi')$. To see why, consider the following reassignment procedure.
Repeatedly until there is no such PM, locate a PM $s_i$ with at least 3 VMs. Then,
locate a PM $s_j$ with one single VM (which exists by the pigeonhole principle). Then, move one VM from $s_i$ to $s_j$.
From Lemma~\ref{l-balance} each movement decreases the power consumed.
Hence, $\pi''$ is still a partition that uses $3m/4$ PMs, each PM has at most 2 VMs assigned,
and $P(\pi'') \leq P(\pi')$.

Then, we know that $P(\pi)$ is not smaller than the power consumption of a partition where exactly $3m/4$ PMs are used and no PM is assigned more than $2$ VMs.
On the other hand, OPT would have assigned each VM to a different PM. Thus, using that $x^* = (b/(\alpha - 1))^{1/\alpha}$, the competitive ratio is

\color{black}

%
%
%\anto{
\begin{eqnarray*}
 \rho &\geq& \frac{(2\beta x^*)^\alpha m/4 + (\beta x^*)^\alpha m/2 + 3mb/4 }{m(\beta x^*)^\alpha + mb}  \\
 &\geq& \frac{(2^{\alpha - 2} +1/2) \beta^\alpha}{\beta^\alpha + (\alpha - 1)} \geq 2^{\alpha - 3} + 1/4,
 \end{eqnarray*}
 %}
where the last inequality follows from $\beta^\alpha \geq (\alpha - 1)$.
Finally, observe that
%\anto{
$2^{\alpha - 3} +1/4 \geq 3^\alpha/(2^{\alpha+2} + \epsilon)$ for $\alpha>1$.
%}
No more VMs arrive in this case.

Let us consider now the the case where \alg assigns the $m$ initial VMs to more than $3m/4$ PMs. Then, after \alg has assigned the first $m$ VMs,
%that arrived, a second batch of VMs arrives. In this second phase, $m/2$ VMs, each with load $2\beta x^*$, arrive.
a second batch of $m/2$ VMs arrive, each VM with load $2\beta x^*$.
Let $\pi$ be the partition output by \alg after this second batch is assigned. If in $\pi$ two of the %new
second batch
VMs are assigned to
the same PM $s_i$, by the pigeonhole principle there is at least one PM $s_j$ with at most load $\beta x^*$. Then, from
Lemma~\ref{l-balance}, the power consumed is reduced if one of the new VMs is moved from $s_i$ to $s_j$. After repeating this process as many times as possible,
a partition $\pi'$ is obtained %in which
where each of the VMs of the second batch is assigned to a different PM,
and $P(\pi') \leq P(\pi)$. Since \alg used more than $3m/4$ PMs in the first batch, in $\pi'$, there are at least
$m/4$ PMs with load $3\beta x^*$. On %its
the other hand, \opt can distribute all the VMs in such a way that each PM has a load of $2\beta x^*$.
%The bound on the ratio then is as follows.
Thus, the bound on the competitive ratio is as follows.
\begin{eqnarray*}
\rho &\geq& \frac{m(3\beta x^*)^\alpha/4}{m(2\beta x^*)^\alpha +  mb} \geq \frac{3^\alpha}{2^{\alpha +2}+ \epsilon},
\end{eqnarray*}
where the last inequality follows from $\epsilon \geq (\alpha - 1)/\beta^\alpha$. \qed
%$\beta^\alpha \geq (\alpha - 1)$.}
\end{proof}
%\end{comment}

%\subsubsection{A Stronger Lower Bound for \MVMA when $m=2$.}
Now, we show a stronger lower bound on the competitive ratio for $(\cdot,2)$-VMA problem. %in this section that the above lower bound can be made stronger for $m=2$.

\begin{theorem}
\label{thm:lb2servers}
%When $m=2$,
There exists an instance of problem $(\cdot,2)$-VMA such that no online algorithm can guarantee a competitive ratio smaller than $3^\alpha/2^{\alpha+1}$.
%There is no online algorithm that achieves a competitive ratio better than $3^\alpha/2^{\alpha+1}$ for $(\cdot,2)$-VMA problem.
\end{theorem}

\begin{proof}
We prove the result by showing an adversarial arrival of VM.
We evaluate the competitive ratio of any online algorithm \alg with respect to an optimal algorithm \opt that knows the future VM arrivals.
The adversarial arrival follows. In a first phase two VM $d_1$ and $d_2$ arrive, with loads $\l(d_1)=\l(d_2)=6x^*$ (Recall from Section \ref{sec:prelim} that $x^* = \left(b/(\alpha-1)\right)^{1/\alpha}$).

If \alg assigns both VMs to the same PM, the power consumed will be $(12x^*)^\alpha + b$, whereas \opt would assign them to different PMs, with a power consumption of $2((6x^*)^\alpha+b)$. Hence, the ratio $\rho$ would be
\begin{eqnarray*}
\rho &=& \frac{(12x^*)^\alpha + b}{2((6x^*)^\alpha+b)} > \frac{12^\alpha}{2(6^\alpha+\alpha-1)} \\
&>& \frac{12^\alpha}{2(6^\alpha+2^\alpha)} = \frac{6^\alpha}{2(3^\alpha+1)},
\end{eqnarray*}
where the first inequality follows from $\alpha>1$ and the second from
$\alpha-1< 2^\alpha$ for any $\alpha>1$.
It is enough to prove that $6^\alpha/(2(3^\alpha+1)) \geq \left(3/2\right)^{\alpha}/2$, or equivalently
$4^\alpha \geq 3^\alpha+1$,
%\begin{align*}
%\frac{6^\alpha}{2(3^\alpha+1)} &\geq \frac{1}{2}\left(\frac{3}{2}\right)^{\alpha}\\
%\frac{6^\alpha}{3^\alpha+1} &\geq \left(\frac{3}{2}\right)^{\alpha}\\
%4^\alpha &\geq 3^\alpha+1.
%\end{align*}
which is true for any $\alpha>1$.
Then, there are no new VM arrivals.

If, otherwise, \alg assigns each VM $d_1$ and $d_2$ to a different PM, then a third VM $d_3$ arrives, with load $\l(d_3)=12x^*$.
Then, \alg must assign it to one of the PMs. Independently of which PM is used, the power consumption of the final configuration
is $(18x^*)^\alpha + (6x^*)^\alpha + 2b$. On its side, \opt assigns $d_1$ and $d_2$ to one PM, and $d_3$ to the other, with a power
consumption of $2((12x^*)^\alpha+b)$. Hence, the competitive ratio $\rho$ is
\begin{eqnarray*}
\rho &=& \frac{(18x^*)^\alpha + (6x^*)^\alpha + 2b}{2((12x^*)^\alpha+b)} > \frac{18^\alpha + 6^\alpha}{2(12^\alpha+\alpha-1)} \\
&>& \frac{18^\alpha + 6^\alpha}{2(12^\alpha+4^\alpha)} \geq \left(3/2\right)^{\alpha}/2,
\end{eqnarray*}
where the first inequality follows from $\alpha>1$, the second from
$\alpha-1< 4^\alpha$ for any $\alpha>1$, and the third from
$(9^\alpha + 3^\alpha)/(6^\alpha+2^\alpha) \geq \left(3/2\right)^{\alpha}$,
%\begin{align*}
%\frac{9^\alpha + 3^\alpha}{2(6^\alpha+2^\alpha)} &\geq \frac{1}{2}\left(\frac{3}{2}\right)^{\alpha}\\
%\frac{9^\alpha + 3^\alpha}{6^\alpha+2^\alpha} &\geq \left(\frac{3}{2}\right)^{\alpha}.
%\end{align*}
what can be checked to be true. Then, there are no new VM arrivals and the claim follows. \qed
\end{proof}

	\subsection{Upper Bounds}
	\label{sec:online-upper}
	%!TEX root = main.tex

%\subsubsection{Upper Bounds for \MVMA and \CMVMA.}
Now, we study upper bounds for  \VMA, \CVMA, and $(\cdot,2)$-VMA problems. We start giving an online VMA algorithm that can be used in \VMA and \CVMA problems. % case when the number of PMs is unbounded.
%, i.e., $m=\infty$.
%
%We denote $z =\frac{ \min \left\{x^*, C\right\}}{2}$, so that $z = x^*/2$ for \VMA and \CVMA when $x^* < C$, while $z=C/2$ for \CVMA when $x^* \geq C$.
The algorithm uses the load of the new revealed VM in order to decide the PM
where it will be assigned.
%\anto{I would assume that $z = \max\{x^*, C\}$}.
If the load of the revealed VM is strictly larger than $\min\{x^*, C\}/2$, the algorithm assigns this VM to a new PM without any other VM already assigned to it. %(For \CVMA with $x^* \geq C$, there will be no VM larger than $z$.)
Otherwise, the algorithm schedules the revealed VM to any loaded PM whose current load is smaller or equal than $\frac{\min\{x^*, C\}}{2}$. Hence, when this new VM is assigned, the load of this PM remains smaller than $\min\{x^*, C\}$.
If there is no such loaded PM, the revealed VM is assigned to a new PM.
Note that, since the case under consideration assumes the existence of an unbounded number of PMs, there exists always one new PM. % whose current load is smaller than $z$.
A detailed description of this algorithm is shown in Algorithm \ref{alg:inftym}. As before, $A_j$ denotes the set of VMs assigned to PM $s_j$ at a given time.
\begin{algorithm}

%\DontPrintSemicolon
\For{ each VM $d_i$}{
\eIf{$\l(d_i) >  \frac{\min\{x^*, C\}}{2}$}
{$d_i$ is assigned to a new PM\:}
%{$d_i$ is assigned to the PM $s_j$ such that $\l(A_k) \leq \l(A_j) < x^*$ for all $k$\:}
{$d_i$ is assigned to any loaded PM $s_j$ where $\ell(A_j) \leq \frac{\min\{x^*, C\}}{2}$. % such that $\l(A_j)  +d_i \leq \min\{x^*, C\}$. 
If such loaded PM does not exist, $d_i$ is assigned to a new PM\:}

}
\caption{Online algorithm for \VMA and \CVMA problems.}
\label{alg:inftym}
\end{algorithm}

We prove the approximation ratio of Algorithm \ref{alg:inftym} in the following two theorems.
\begin{theorem}
\label{theo:UBunboundedPMs}
%%%%%%%%% replaced \sum_{d \in S}\l(d) with \ell(D_s)
%For \VMA and \CVMA, Algorithm \ref{alg:inftym} achieves a competitive ratio of $1$ if no VM $d_i$ has $\l(d_i) < x^*$, and of  $\left(1-\frac{1}{\alpha}\left(1-\frac{1}{2^\alpha}\right)\right)\left(2 + \frac{ x^*}{\sum_{d \in S}\l(d)}\right)$
%%$2^{\alpha - 1} + x^*\big/\sum_{d_i:\l(d_{i})< x^*}\l(d_{i})$
%, otherwise.
There exists an online algorithm for \VMA and \CVMA when $x^* < C$ that achieves the following competitive ratio: 
\begin{eqnarray*}
 \rho &= 1, \textrm{ if no VM $d_i$ has load such that $\l(d_i) < x^*$,}\\
 \rho &\leq  \left(1-\frac{1}{\alpha}\left(1-\frac{1}{2^\alpha}\right)\right)\left(2 + \frac{ x^*}{\ell(D_s)}\right), \textrm{ otherwise.}
\end{eqnarray*}
%For \VMA and \CVMA when $x^* < C$, Algorithm \ref{alg:inftym} achieves a competitive ratio of $1$ if no VM $d_i$ has $\l(d_i) < x^*$, and of  $\left(1-\frac{1}{\alpha}\left(1-\frac{1}{2^\alpha}\right)\right)\left(2 + \frac{ x^*}{\ell(D_s)}\right)$
%$2^{\alpha - 1} + x^*\big/\sum_{d_i:\l(d_{i})< x^*}\l(d_{i})$
%, otherwise.
\end{theorem}

%\begin{comment}
\begin{proof}
We proceed with the analysis of the competitive ratio of Algorithm~\ref{alg:inftym} shown above.
Let us first consider an optimal algorithm, that is, an algorithm that gives an optimal solution for any instance.
%Now consider any instance of VMA problem.
Let us denote by $\pi^*$ the optimal solution obtained by the optimal algorithm, and $A_i$ the load assigned to PM $s_i$ in that solution, for a particular instance of VMA problem.
%\footnote{In order to keep a simple and understandable notation, we omit notation for the particular instance. Though, we are considering a particular instance of VMA problem.}.
Furthermore, load $A_i$ is decomposed in $d_{i_{1}}, d_{i_{2}}, \ldots, d_{i_{k_i}}$, where each $d_{i_j}$ is a VM that $\pi^*$ assigns to $s_i$.
Using simple algebra, it holds:
$$
%\begin{eqnarray*}
f(\l(A_i)) = \frac{f(\l(A_i))}{\l(A_i)}(\l(d_{i_{1}}) + \l(d_{i_{2}}) + \cdots + \l(d_{i_{k_i}})).
%&=&\sum \frac{f(\l(A_i))}{\l(A_i)} + \sum \frac{f(\l(A_i))}{\l(A_i)}
%\end{eqnarray*}
$$
It is possible now to split the set $A_i$ in two sets, one with those VMs assigned to $s_i$ whose load is strictly smaller than $x^*$ and a second set that contains those VMs assigned to $s_i$ whose load is bigger than $x^*$.
In terms of notation, we say that $A_i$ is split in $B_i$ and $S_i$ (where $B$ stands for Big loads and $S$ stands for Small loads).
Therefore, it also holds:
$$
%\begin{eqnarray*}
f(\l(A_i)) %&=& \frac{f(\l(A_i))}{\l(A_i)}(\l(d_{i_{1}}) + \l(d_{i_{2}}) + \cdots + \l(d_{i_{k_i}}))\\
=\sum_{d_{i_j} \in B_i} \frac{f(\l(A_i))}{\l(A_i)}\l(d_{i_j}) + \sum_{d_{i_j} \in S_i} \frac{f(\l(A_i))}{\l(A_i)}\l(d_{i_j}).
%\end{eqnarray*}
$$

On the other hand, by definition of $x^*$, it holds that: $$f(\l(A_i))/ \l(A_i) \geq f(x^*)/ x^*$$ for all $i$ (indeed, for any load).
%\anto{For the case $x^* \geq C$, we have $f(\l(A_i))/ \l(A_i) \geq f(C)/ C$. Thus, $f(\l(A_i))/ \l(A_i) \geq f(z)/ z$.}
Moreover, if a PM has been assigned with a load $\l(d_{i_j})$ bigger than $x^*$, it also holds that $f(\l(A_i))/ \l(A_i) \geq f(\l(d_{i_j}))/\l(d_{i_j})$.
Hence, we obtain the following inequality:
$$
f(\l(A_i)) \geq \sum_{d_{i_j} \in B_i}f( \l(d_{i_j})) + \sum_{d_{i_j} \in S_i} \frac{f(x^*)}{x^*}\l(d_{i_j}).
$$

In order to lower bound the power consumption of the solution $\pi^*$, we plug the above inequality into the corresponding equation:
\begin{eqnarray*}
P(\pi^*) &=& \sum_{A_i\neq \emptyset} f(\l(A_i)) \\
&\geq&  \sum_{A_i\neq \emptyset}\sum_{d_{i_j} \in B_i}f(\l(d_{i_j}))+ \frac{f(x^*)}{x^*} \sum_{A_i\neq \emptyset}\sum_{d_{i_j} \in S_i}\l(d_{i_j}),
\end{eqnarray*}
or, equivalently expressed in more compact notation:
$$
P(\pi^*) \geq \sum_{d_i:\l(d_{i})\geq x^*}f(\l(d_{i})) + \frac{f(x^*)}{x^*} \sum_{d_i:\l(d_{i})< x^*}\l(d_{i}).
$$

Consider now Algorithm  \ref{alg:inftym}. Let us denote by $\pi$ a solution that Algorithm  \ref{alg:inftym} gives for a particular instance.
Also, let us denote by $\hat{A}_i$ the load assigned by Algorithm \ref{alg:inftym} to PM $s_i$.
 Note that due to the design of the algorithm,
after the last VM has been assigned, either there is only one loaded PM whose current load is smaller than $x^*/2$,
or every loaded PM has a load at least $x^*/2$. We study these two cases separately.  \\
%\begin{itemize}
%\item { 
\emph{Case 1:} $\l(\hat{A}_i) \geq x^*/2$ for all $i$.
In this case, in a solution provided by $\pi$ there are PMs with two types of load:
those that are loaded with one VM whose load is no smaller than $x^*$,
and those that are loaded with VMs whose load is strictly smaller than $x^*$, nonetheless, their total load is bigger than $x^*/2$.
Note that due to the design of the algorithm, none of the PMs in the second group has a load bigger than $x^*$.
Let us denote by $B$ the set of VMs with load at least $x^*$,
%and $S$ the set of VMs with load less than $x^*$.
and $D_s$ the set of VMs with load less than $x^*$.
Therefore, it holds:
\begin{eqnarray*}
P(\pi) &=& \sum_{d \in B}f(\l(d)) + \sum_{\frac{x^*}{2} \leq \l(\hat{A}_i) \leq x^*} f(\l(\hat{A}_i)) \\
&\leq& \sum_{d \in B}f(\l(d)) + \frac{f(\frac{x^*}{2})}{\frac{x^*}{2}} \ell(D_s).
\end{eqnarray*}
{%\color{red}
%Let us denote by $B$ the set of VMs with load at least $x^*$, and $S$ the set of VMs
%with load less than $x^*$.
Computing the ratio $\rho$ between $P(\pi)$ and $P(\pi^*)$, we obtain the following inequality:
\begin{eqnarray*}
\rho &\leq& \frac{\sum_{d \in B}f(\l(d)) + \frac{f(\frac{x^*}{2})}{\frac{x^*}{2}} \ell(D_s)}{\sum_{d \in B}f(\l(d)) + \frac{f(x^*)}{x^*} \ell(D_s)} \label{e-eq1} \leq \frac{\frac{f(\frac{x^*}{2})}{\frac{x^*}{2}} \ell(D_s)}{\frac{f(x^*)}{x^*} \ell(D_s)} \\ %\nonumber \\
&=& 2\frac{f(\frac{x^*}{2})}{f(x^*)} = 2\left(1-\frac{1}{\alpha}\left(1-\frac{1}{2^\alpha}\right)\right).
%\: \leq \: 2^\alpha.
\end{eqnarray*}
}
%\item {
\emph{Case 2:} there exists $s_i$ such that $\l(\hat{A}_i) < x^*/2$.
In this case, $\pi$ gives solutions with three types of loaded PMs:
those that are loaded with one VM whose load is bigger than $x^*$,
those that are loaded with VMs whose load is strictly smaller than $x^*$, but which total load is at least $x^*/2$,
and one PM whose total load is is strictly smaller than $x^*/2$.
Let us denote such a PM by $s'$.
Therefore, it holds:
\begin{align*}
& P(\pi) =  \sum_{d \in B}f(\l(d)) + \sum_{\frac{x^*}{2} \leq \l(\hat{A}_i) \leq x^*} f(\l(\hat{A}_i)) +  f(\l(\hat{A}_{s'}))\\
& \leq \sum_{d \in B}f(\l(d)) + \frac{f(\frac{x^*}{2})}{\frac{x^*}{2}} \Big(\ell(D_s) - \l(\hat{A}_{s'})\Big)+ f(\l(\hat{A}_{s'}))\\
& =\sum_{d \in B}f(\l(d)) + \frac{f(\frac{x^*}{2})}{\frac{x^*}{2}} \Big(\ell(D_s) - \l(\hat{A}_{s'})\Big) + \l(\hat{A}_{s'})^\alpha + b.
\end{align*}
Let us denote the latter expression by $\Pi(\pi)$.
Computing the ratio $\rho$ between $P(\pi)$ and $P(\pi^*)$, we obtain the following inequality:
\begin{eqnarray*}
 %\frac{P(\pi)}{P(\pi^*)}
 \rho
%&\leq& \frac{\sum_{d \in B}f(\l(d)) + \frac{f(2x^*)}{2x^*} (\ell(D_s) - \hat{A}_{s'}) +
%\l(\hat{A}_{s'})^\alpha + b}{\sum_{d \in B}f(\l(d)) + \frac{f(x^*)}{x^*} \ell(D_s)} \label{e-eq2} \\
&\leq& \frac{\Pi(\pi)}{\sum_{d \in B}f(\l(d)) + \frac{f(x^*)}{x^*} \ell(D_s)} \label{e-eq2} \\
&\leq& 2\left(1-\frac{1}{\alpha}\left(1-\frac{1}{2^\alpha}\right)\right) + \frac{ \l(\hat{A}_{s'})^\alpha - \l(\hat{A}_{s'})\frac{f(\frac{x^*}{2})}{\frac{x^*}{2}}+ b}{ \frac{f(x^*)}{x^*} \ell(D_s)} \\
&\leq& 2\left(1-\frac{1}{\alpha}\left(1-\frac{1}{2^\alpha}\right)\right) + \frac{ \l(\hat{A}_{s'})^\alpha + b}{ \frac{f(x^*)}{x^*} \ell(D_s)} \nonumber \\
&\leq& 2\left(1-\frac{1}{\alpha}\left(1-\frac{1}{2^\alpha}\right)\right) + \frac{ (\frac{x^*}{2})^\alpha + b}{ \frac{f(x^*)}{x^*} \ell(D_s)}\\
& = & \left(1-\frac{1}{\alpha}\left(1-\frac{1}{2^\alpha}\right)\right)\left(2 + \frac{ x^*}{\ell(D_s)}\right). \nonumber
\end{eqnarray*}
%\end{itemize}
Since $ x^* / \ell(D_s)$ is always positive, the competitive ratio of Algorithm \ref{alg:inftym} is equal to $2^{\alpha - 1} +  x^* / \ell(D_s)$.
Observe that, when no VM
$d$ has load $\l(d)<x^*$, i,e., $S=\emptyset$,
$P(\pi)$ and $P(\pi^*)$ are equal. Hence, the competitive ratio is $1$.
%equations (\ref{e-eq1}) and (\ref{e-eq2}) become $\frac{P(\pi)}{P(\pi^*)}
%\leq 1$.
\qed
\end{proof}
%\end{comment}

\begin{theorem}
There exists an online algorithm for  \CVMA when $x^* \geq C$ that achieves competitive ratio $\rho \leq \frac{2b}{C}\left( 1 + \frac{1}{(\alpha-1)2^{\alpha}} \right)\left(2 + \frac{C}{\ell(D)}\right)$.
%For \CVMA when $x^* \geq C$, Algorithm \ref{alg:inftym} can achieve a competitive ratio of  $\frac{2b}{C}\left( 1 + \frac{1}{(\alpha-1)2^{\alpha}} \right)\left(2 + \frac{C}{\ell(D)}\right)$.
\end{theorem}
\begin{proof}
We proceed with the analysis of the competitive ratio of Algorithm~\ref{alg:inftym} in the case when $x^* \geq C$.
The analysis uses the same technique used in the proof for the previous theorem. Hence, we just show the difference.

On the one hand, when $x^* \geq C$, it holds that $f(\l(A_i))/ \l(A_i) \geq f(C)/ C$ due to the fact that $f(x)/x$ is monotone decreasing in interval $(0, C]$. It is also obvious that all the PMs will be loaded no more $C$. As a result, the optimal power consumption for \CVMA can be bounded by
$$
P(\pi^*) \geq \frac{f(C)}{C} \ell(D).
$$

On the other hand, the solution given by Algorithm~\ref{alg:inftym} can also be upper bounded. 
%Using the same notation and technique that we used to prove Theorem \ref{theo:UBunboundedPMs}, w
We consider the following two cases. \\
%\begin{itemize}
%\item {
\emph{Case 1:} $\l(\hat{A}_i) \geq C/2$ for all $i$.
In this case, every PM will be loaded between $C/2$ and $C$. Consequently,
\begin{eqnarray*}
P(\pi) = \sum_{\frac{C}{2} \leq \l(\hat{A}_i) \leq C} f(\l(\hat{A}_i))
\leq \frac{f(\frac{C}{2})}{\frac{C}{2}} \ell(D).
\end{eqnarray*}
The competitive ratio $\rho$ then satisfies
\begin{eqnarray*}
\rho \leq \frac{\frac{f(\frac{C}{2})}{\frac{C}{2}} \ell(D)}{\frac{f(C)}{C} \ell(D)}
\: = \: 2\frac{f(\frac{C}{2})}{f(C)} \leq \frac{2b}{C}\left( 1 + \frac{1}{(\alpha-1)2^{\alpha}} \right).
\end{eqnarray*}
%\item {
\emph{Case 2:}  there exists $s_i$ such that $\l(\hat{A}_i) < C/2$. In this case, it holds: 
\begin{eqnarray*}
P(\pi) &=& \sum_{\frac{C}{2} \leq \l(\hat{A}_i) \leq C} f(\l(\hat{A}_i)) +  f(\l(\hat{A}_{s'}))\\
&\leq& \frac{f(\frac{C}{2})}{\frac{C}{2}} \Big(\sum_{d_i:\l(d_{i}) \leq  C}\l(d_{i}) - \l(\hat{A}_{s'})\Big)+ f(\l(\hat{A}_{s'}))\\
&=& \frac{f(\frac{C}{2})}{\frac{C}{2}} \Big(\ell(D) - \l(\hat{A}_{s'})\Big) + \l(\hat{A}_{s'})^\alpha + b.
\end{eqnarray*}
The competitive ratio $\rho$ then satisfies
\begin{eqnarray*}
 \rho
&\leq& \frac{P(\pi)}{\frac{f(C)}{C} \ell(D)}
\leq \frac{2b}{C}\left( 1 + \frac{1}{(\alpha-1)2^{\alpha}} \right) + \\
&+& \frac{ \l(\hat{A}_{s'})^\alpha - \l(\hat{A}_{s'})\frac{f(\frac{C}{2})}{\frac{C}{2}}+ b}{ \frac{f(C)}{C} \ell(D)} \\
&\leq&  \frac{2b}{C}\left( 1 + \frac{1}{(\alpha-1)2^{\alpha}} \right)+ \frac{ \l(\hat{A}_{s'})^\alpha + b}{ \frac{f(C)}{C} \ell(D)} \nonumber \\
&\leq& \frac{2b}{C}\left( 1 + \frac{1}{(\alpha-1)2^{\alpha}} \right) + \frac{ (\frac{C}{2})^\alpha + b}{ \frac{f(C)}{C} \ell(D)}\\
& = & \frac{2b}{C}\left( 1 + \frac{1}{(\alpha-1)2^{\alpha}} \right)\left(2 + \frac{C}{\ell(D)}\right). \nonumber
\end{eqnarray*}
%\end{itemize}
\qed
\end{proof}

\subsubsection{Upper Bounds for $(\cdot,2)$-VMA problem} We now present an algorithm (detailed in Algorithm~\ref{m2alg}) for $(\cdot,2)$-VMA problem and show an upper bound on its competitive ratio.
%The algorithm is online, that is, the VMs are revealed to the algorithm one by one. 
$A_1$ and $A_2$ are the sets of VMs
assigned to PMs $s_1$ and $s_2$, respectively, at any given time.

\begin{algorithm}
%\DontPrintSemicolon
\For{ each VM $d_i$}{
\eIf{ $\l(d_i)+\l(A_1)\leq \left(b/(2^\alpha-2)\right)^{1/\alpha}$ {\bf or} $\l(A_1)\leq\l(A_2)$}
{$d_i$ is assigned to $s_1$\;}
{$d_i$ is assigned to $s_2$\;}
}
\caption{Online algorithm for $(\cdot,2)$-VMA.}
\label{m2alg}
\end{algorithm}

We prove the approximation ratio of Algorithm~\ref{m2alg} in the following theorem.

\begin{theorem}
\label{theo:UB2PMs}
%For a system where $m=2$, 
There exists an online algorithm for $(\cdot,2)$-VMA that achieves the following competitive ratios.
\begin{eqnarray*}
\rho = 1, &\textrm{  for $\l(D)\leq \left(\frac{b}{2^\alpha-2}\right)^{1/\alpha}$},\\
\rho \leq \max\left\{2,\left(\frac{3}{2}\right)^{\alpha-1}\right\}, &\textrm{  for $\l(D)> \left(\frac{b}{2^\alpha-2}\right)^{1/\alpha}$.}
\end{eqnarray*}
%\begin{align*}
%\rho &= 1, &\textrm{  for $\l(D)\leq \left(\frac{b}{2^\alpha-2}\right)^{1/\alpha}$},\\
%\rho &\leq \max\left\{2,\left(\frac{3}{2}\right)^{\alpha-1}\right\}, &\textrm{  for $\l(D)> \left(\frac{b}{2^\alpha-2}\right)^{1/\alpha}$.}
%\end{align*}
\end{theorem}

%\begin{comment}
\begin{proof}
Consider Algorithm~\ref{m2alg} shown above.
If $\l(D)\leq \left(b/(2^\alpha-2)\right)^{1/\alpha}$, then the competitive ratio is $1$ as we show. Algorithm~\ref{m2alg} assigns all the VMs to PM $s_1$. On the other hand, the optimal offline algorithm also assigns all the VMs to one PM. To prove it, it is enough to show that
%\begin{align}
%\l(D)^\alpha+b &< \l(A_1)^\alpha + \l(A_2)^\alpha + 2b\nonumber\\
%\l(D)^\alpha+b &< 2\left(\frac{\l(D)}{2}\right)^\alpha + 2b\nonumber\\
%\l(D)^\alpha &< \frac{b 2^\alpha}{2^\alpha-2}\nonumber\\
%\l(D) &< 2\left(\frac{b }{2^\alpha-2}\right)^{1/\alpha}.\label{opt1server}
%\end{align}
$\l(D)^\alpha+b < \l(A_1)^\alpha + \l(A_2)^\alpha + 2b$. Using that $\l(A_1)^\alpha + \l(A_2)^\alpha > 2\left(\l(D)/2\right)^\alpha$ and manipulating, it is enough to prove
$\l(D) < 2\left(b / (2^\alpha-2)\right)^{1/\alpha}$.
This is true for $\l(D)\leq \left(b/(2^\alpha-2)\right)^{1/\alpha}$.

%--------------------------------

We consider now the case $\left(b/(2^\alpha-2)\right)^{1/\alpha} < \l(D) < 2\left(b/(2^\alpha-2)\right)^{1/\alpha}$. Within this range, for the optimal algorithm is still better to assign all VMs to one PM, as shown. % by Inequality~\ref{opt1server}.
Then, the competitive ratio $\rho$ is
\begin{align}
\rho
= \frac{\l(A_1)^\alpha+\l(A_2)^\alpha+2b}{\l(D)^\alpha+b}
\leq \frac{\l(D)^\alpha+2b}{\l(D)^\alpha+b}
< 2.\label{ratio2}
\end{align}

%--------------------------------

%Finally,
Consider any given step after $\l(D) \geq 2\left(b/(2^\alpha-2)\right)^{1/\alpha}$.
Within this range, the optimal algorithm may assign the VMs to one or both PMs.
If the optimal algorithm assigns to one PM, Inequality~\ref{ratio2} applies.
Otherwise, the competitive ratio $\rho$ is
%\begin{align*}
\begin{eqnarray*}
\rho&=& \frac{\l(A_1)^\alpha+\l(A_2)^\alpha+2b}{2(\l(D)/2)^\alpha+2b}
\leq 2^{\alpha-1} \frac{\l(A_1)^\alpha+\l(A_2)^\alpha}{\l(D)^\alpha} \\
&=& 2^{\alpha-1} \frac{\l(A_1)^\alpha/\l(A_2)^\alpha+1}{(\l(A_1)/\l(A_2)+1)^\alpha}.
\end{eqnarray*}
%\end{align*}
Then, in order to obtain a ratio at most $x^\alpha/2$, where $x$ will be set later, it is enough to guarantee
\begin{align*}
2^{\alpha-1} \frac{\l(A_1)^\alpha/\l(A_2)^\alpha+1}{(\l(A_1)/\l(A_2)+1)^\alpha}
&\leq  \frac{x^\alpha}{2}\\
\frac{(\l(A_1)/\l(A_2))^\alpha+1}{(\l(A_1)/\l(A_2)+1)^\alpha}
&\leq \left(\frac{x}{2}\right)^\alpha.
\end{align*}
Without loss of generality, assume $\l(A_1)\leq \l(A_2)$. This implies that $(\l(A_1)/\l(A_2))^\alpha \leq \l(A_1)/\l(A_2)$.
Then, it is enough to have
\begin{align*}
\frac{\l(A_1)/\l(A_2)+1}{(\l(A_1)/\l(A_2)+1)^\alpha}
&\leq \left(\frac{x}{2}\right)^\alpha.
\end{align*}
Let us now define $\l(A_1)+\l=\l(A_2)$ for some $\l\geq 0$. Manipulating and replacing, it is enough to show
\begin{align}
%\left(\frac{2}{x}\right)^{\alpha/(\alpha-1)} -1
%&\leq
%\frac{\l(A_1)}{\l(A_2)}\nonumber\\
%\frac{\l(A_2)}{\l(A_1)}
%&\leq
%\frac{1}{\left(2/x\right)^{\alpha/(\alpha-1)} -1}\nonumber\\
%\frac{\l}{\l(A_1)}
%&\leq
%\frac{1}{\left(2/x\right)^{\alpha/(\alpha-1)} -1} - 1\nonumber\\
\frac{\l}{\l(A_1)}
&\leq
\frac{2-\left(2/x\right)^{\alpha/(\alpha-1)}}{\left(2/x\right)^{\alpha/(\alpha-1)} -1}. \label{eq:condratio}
\end{align}

If Inequality~\ref{eq:condratio} holds the theorem is proved. Otherwise, the following claim is needed.
\begin{claim}
\label{claim:max}
If $\l(D) \geq 2\left(b/(2^\alpha-2)\right)^{1/\alpha}$, then there must exist a VM $d_i$ in $D$ such that $\l(d_i)\geq|\l(A_2)-\l(A_1)|$.
\end{claim}
\begin{proof}
If $\l(A_2)=\l(A_1)$ the claim follows trivially. Assume that $\l(A_2)\neq\l(A_1)$.
Consider any given time when $\l(D) \geq 2\left(b/(2^\alpha-2)\right)^{1/\alpha}$.
For the sake of contradiction, assume that for all $d_i\in D$ it is $\l(d_i)<|\l(A_2)-\l(A_1)|$.
Let $d_1,d_2,\dots,d_r$ be the order in which the VMs were revealed to Algorithm~\ref{m2alg}.
And let the respective sets of VMs be called $D_i=\{d_j|j\in [1,i]\}$, that is $D_r=D$.
Given that $\l(D) \geq 2\left(b/(2^\alpha-2)\right)^{1/\alpha} > \left(b/(2^\alpha-2)\right)^{1/\alpha}$, the VM $d_r$ was assigned to the PM with smaller load. Then, either $\l(d_r)\geq |\l(A_2)-\l(A_1)|$ which would be a contradiction, or if $\l(d_r) < |\l(A_2)-\l(A_1)|$ the
PM with the smaller load before and after assigning $d_r$ is the same.
The argument can be repeated iteratively backwards for each $d_{r-1}$, $d_{r-2}$, etc. until, for some $j\in [1,r)$, either it is $\l(d_j)\geq |\l(A_2)-\l(A_1)|$ reaching a contradiction, or the total load is $\l(D_j)<\left(b/(2^\alpha-2)\right)^{1/\alpha}$.
If the latter is the case, we know that for $i\in [1,j]$ every $d_i$ was assigned to $s_1$.
Recall that for $i\in(j,r]$ each $d_i$ was assigned to the same PM. And, given that $d_{j+1}$ is the first VM for which the total load is at least $\left(b/(2^\alpha-2)\right)^{1/\alpha}$, that PM is $s_2$.
But then, we have $\l(A_2)< \l(A_1)< \left(b/(2^\alpha-2)\right)^{1/\alpha}$, which is a contradiction with the assumption that $\l(D) \geq 2\left(b/(2^\alpha-2)\right)^{1/\alpha}$.
\end{proof}

Using Claim~\ref{claim:max} we know that there exists a $d_i$ in the input such that
\begin{align*}
\l(d_i) \geq \l > \l(A_1)
\frac{2-\left(2/x\right)^{\alpha/(\alpha-1)}}{\left(2/x\right)^{\alpha/(\alpha-1)} -1}.
\end{align*}
From the latter, it can be seen that if $x\geq 2 (3/4)^{\frac{\alpha-1}{\alpha}}$, then we have that $\l>2\l(A_1)$.
Then, the competitive ratio $\rho$ is
\begin{eqnarray*}
%\begin{align*}
\rho &=& \frac{\l(A_1)^\alpha+(\l(A_1)+\l)^\alpha+2b}{(2\l(A_1))^\alpha+\l^\alpha+2b} \\
&\leq& \frac{\l(A_1)^\alpha+(\l(A_1)+\l)^\alpha}{(2\l(A_1))^\alpha+\l^\alpha}.
%\end{align*}
\end{eqnarray*}
Using calculus, this ratio is maximized for $\l=2\l(A_1)$ for $\l\geq 2\l(A_1)$. Then, we have
%\begin{align*}
%\rho
%\leq \frac{1+3^\alpha}{2\cdot2^\alpha}.
%\end{align*}
$\rho
\leq (1+3^\alpha)/(2\cdot2^\alpha)$.
Then, in order to obtain a ratio at most $x^\alpha/2$, it is enough to guarantee
%\begin{align*}
%\frac{1+3^\alpha}{2\cdot2^\alpha} &\leq \frac{x^\alpha}{2}\\
%x &\geq \left(\frac{1+3^\alpha}{2^\alpha}\right)^{1/\alpha}.
%\end{align*}
$(1+3^\alpha)/(2\cdot2^\alpha) \leq x^\alpha/2$
which yields
$x \geq \left((1+3^\alpha)/2^\alpha\right)^{1/\alpha}$.

Given that, for any $\alpha\geq1$, it holds:
%\begin{align*}
%2 (3/4)^{\frac{\alpha-1}{\alpha}} &\geq \left(\frac{1+3^\alpha}{2^\alpha}\right)^{1/\alpha}\\
%2^\alpha (3/4)^{\alpha-1} &\geq \frac{1+3^\alpha}{2^\alpha}\\
%4\cdot 3^{\alpha-1} &\geq 1+3^\alpha\\
%\frac{4}{3} 3^{\alpha} &\geq 1+3^\alpha\\
%\frac{1}{3} 3^{\alpha} &\geq 1\\
%3^{\alpha} &\geq 3.
%\end{align*}
$$2 (3/4)^{1-1/\alpha} \geq \left((1+3^\alpha)/2^\alpha\right)^{1/\alpha}.$$
Then, the competitive ratio is
%\begin{align*}
%\rho &\leq (2 (3/4)^{\frac{\alpha-1}{\alpha}})^\alpha/2\\
%&= 2^\alpha (3/4)^{\alpha-1}/2\\
%&= 2^{\alpha-1} (3/4)^{\alpha-1}\\
%&= (3/2)^{\alpha-1}.
%\end{align*}
$\rho
\leq (2 (3/4)^{1-1/\alpha})^\alpha/2
= (3/2)^{\alpha-1}$.
\qed
\end{proof}
%\end{comment}

%\section{Experimental Evaluation}
%\label{sec:exp}
%\input{Simulations}

\section{Discussion}
\label{sec:discussion}
%In order to apply our results to production environments, there are still some practical issues that need to be handled. We discuss some important ones among them in this section.
We discuss in this section practical issues that must be addressed to apply our results to production environments.

{\bf Heterogeneity of Servers.} For the sake of %tractability,
simplicity, we assume in our model that all servers in a data center are identical. We believe this reasonable, considering that modern data centers are usually built with homogeneous commodity hardware. Nevertheless, the proposed model and derived results are also %friendly with 
amenable to heterogeneous data center environments. 
In a heterogeneous data center, servers can be categorized into several groups with identical servers in each group. 
%Each group of servers is used for running several designated types of applications according to the servers' hardware configuration and applications' characteristics. The assignment of applications to a server group is out of the scope of this paper. However, when this application-to-server group mapping has been done, the proposed VMA model can be perfectly applied to the assignment problem of allocating tasks from the designated types of applications (especially CPU intensive ones) to each group of servers. Therefore, the approximation results we derive in this paper can be served as useful guidelines for energy-efficient task assignment in real data centers, regardless of the homogeneity of servers.
Then, different types of applications can be assigned to server groups according to their resource requirements.
The VMA model presented here can be applied to the assignment problem of allocating tasks from the designated types of applications (especially CPU-intensive ones) to each group of servers. The approximation results we derive in this paper can be then combined with server-group assignment approximation bounds (out of the scope of this paper) for energy-efficient task assignment in real data centers, regardless of the homogeneity of servers.

{\bf Consolidation.} Traditionally, consolidation has been understood as a bin packing problem \cite{MSvector, wangconsolidating}, 
%where VMs where piled up in PMs and we tried to minimize the number of active PMs. 
where VMs are assigned to PMs attempting to minimize the number of active PMs.
However, the results we derived in this paper, as well as the %ones from 
results in~\cite{eenergy}, 
%present evidence enough to think that this is not optimal. 
show that such approach is not energy-efficient.
Indeed, we showed that PM's should be loaded up to $x^*$ to reduce energy consumption, even 
%when this supposes to have more active PMs.
if this requires having more active PMs.

{\bf VM arrival and departure.} When a new VM arrives to the system, or an assigned VM departs, 
%the set of VM demands will be changed. In order to maintain the efficiency of the assignment, adjustments have to be made according to the new set of VMs. 
adjustments to the assignment may improve energy efficiency.
%Given that VM migration is now becoming more and more costless and can be accomplished instantaneously, 
Given that the cost of VM migration is nowadays decreasing dramatically,
our offline positive results can also be accommodated by reassigning VMs %once
whenever the set of VM demands changes. 
%Another option would be dividing time into equal-length slots. Instead of reassigning VMs at every VM arrival or departure, the newly arrived VM demands are buffered until the next slot. At the beginning of each slot, the VMs including the running ones from previous slots and the ones that arrive in the previous slot are rescheduled and redistributed by running an offline approximation algorithm. 
Should the cost of migration be high to reassign after each VM arrival or departure, time could be divided in epochs buffering newly arrived VM demands until the beginning of the next epoch, when all (new and old) VMs would be reassigned (if necessary) running our offline approximation algorithm.

{\bf Multi-resource scheduling.} This work focuses on CPU-intensive jobs (VMs) such as MapReduce-like tasks \cite{mapreduce} which are representative in production datacenters. As the CPU is generally the dominant energy consumer in a server, assigning VMs according to CPU workloads %could result in the most energy-efficient solution. 
entails energy efficiency.
However, there exist types of jobs demanding heavily other computational resources, such as memory and/or storage. Although these resources have limited impact on a server's energy consumption, VMs performance may be degraded if they become the bottleneck resource in the system. In this case, a joint optimization of multiple resources (out of the scope of this paper) is necessary for VMA. %We leave it as future work.

\section{Conclusions}
\label{sec:conclude}
In this paper, we have studied a particular case of the generalized assignment problem with applications to Cloud Computing. We have considered the problem of assigning virtual machines (VMs) to physical machines (PMs) so that the power consumption is minimized, a problem that we call virtual machine assignment (VMA). In our theoretical analysis, we have shown that the decision version of \CMVMA problem is strongly NP-complete. We have shown as well that the \CVMA, \MVMA and \VMA problems are strongly NP-hard. Hence, there is no FPTAS for these optimization problems. We have shown the existence of a PTAS that solves the \VMA and \MVMA offline problems. On the other hand, we have proved lower bounds on the approximation ratio of the \CVMA and \CMVMA problems. With respect to the online version of these problems, we have proved upper and lower bounds on the competitive ratio of the \VMA, \CVMA, \MVMA, and \CMVMA problems.

\newpage

\bibliographystyle{plain}
\bibliography{refs}

\end{document}